\documentclass[aps,amssymb,amsmath,amsfonts,pra]{revtex4}

\usepackage{bm,bbm}

\usepackage{graphicx}
\usepackage{epstopdf}
\usepackage{amsthm}
\usepackage{amsmath}
\usepackage{color}

\usepackage{hyperref}
\usepackage{graphicx}
\usepackage{graphics}

\usepackage[T1]{fontenc}
\usepackage[english]{babel}
\usepackage[utf8]{inputenc}

\newtheorem{theorem}{Theorem}
\newtheorem{proposition}[theorem]{Proposition}
\newtheorem{lemma}[theorem]{Lemma}
\newtheorem{corollary}[theorem]{Corollary}

\renewcommand{\Re}{\mathbf{Re}}

\newcommand{\R}{\mathbb{R}}

\DeclareMathOperator{\id}{Id}

\newcommand{\<}{\langle}
\renewcommand{\>}{\rangle}

\newcommand{\scalar}[2]{\langle #1 | #2 \rangle}

\newcommand{\ket}[1]{| #1 \rangle}
\newcommand{\bra}[1]{\langle #1 |}
\newcommand{\tr}{\mathrm{Tr}}

\newcommand{\C}{\mathbb{C}}

\newcommand{\E}{\mathbb{E}}

\usepackage{amsmath, amssymb, amsthm}

\def\Ex{\mathbb{E}}
\def\er{\mathbb{R}}
\def\ce{\mathbb{C}}
\def\ve{\varepsilon}
\newcommand{\calE}{\mathcal{E}}
\newcommand{\p}{\mathbb{P}}

\newcommand{\Arg}{{\rm Arg\,}}

\newcommand{\constantCzero}{C_{0}}

\newcommand{\constantCa}{C_{1}}

\newcommand{\constantCb}{\alpha}

\newcommand{\constantCc}{C_{3}}

\newcommand{\constantCd}{C_{13}}

\newcommand{\constantCe}{C_{14}}

\newcommand{\constantCf}{C_{2}}

\newcommand{\constantCg}{C_{4}}

\newcommand{\constantCh}{C_{15}}

\newcommand{\constantCi}{C_{16}}

\newcommand{\constantCj}{C_{17}}

\newcommand{\constantCk}{C_{5}}

\newcommand{\constantCl}{C_{6}}

\newcommand{\constantCm}{C_{7}}

\newcommand{\constantCn}{C_{8}}

\newcommand{\constantCo}{C_{9}}

\newcommand{\constantCp}{C_{10}}

\newcommand{\constantCq}{C_{11}}

\newcommand{\constantCr}{C_{12}}

\begin{document}
\title{Asymptotic entropic uncertainty relations}


\author{Rados{\l}aw Adamczak}

\author{Rafa{\l} Lata{\l}a}

\affiliation{Institute of Mathematics, University of Warsaw,
ul. Banacha 2,  PL-02-097 Warsaw, Poland}

\author{Zbigniew Pucha\l{}a}

\affiliation{Institute of Theoretical and Applied Informatics, Polish Academy
of Sciences, ul. Ba\l{}tycka 5, 44-100 Gliwice, Poland}

\affiliation{Institute of Physics, Jagiellonian University, ul.  \L{}ojasiewicza 11, 30-059
Krak{\'o}w, Poland}

\author{Karol \.{Z}yczkowski}

\affiliation{Institute of Physics, Jagiellonian University, ul. \L{}ojasiewicza 11, 30-059
Krak{\'o}w, Poland}

\affiliation{Center for Theoretical Physics, Polish Academy of Sciences, Aleja
Lotnik{\'o}w 32/46, PL-02-668 Warsaw, Poland}


\date{7/10/2015}

\begin{abstract}
We analyze entropic uncertainty relations for two orthogonal measurements on a
$N$-dimensional Hilbert space, performed in two generic bases. It is assumed
that the unitary matrix $U$ relating both bases is distributed according to the
Haar measure on the unitary group. We provide lower bounds on the average
Shannon entropy of probability distributions related to both measurements. The
bounds are stronger than these obtained with use of the entropic uncertainty
relation by Maassen and Uffink, and they are optimal up to additive constants.

We also analyze the case of a large number of measurements and obtain
strong entropic uncertainty
relations which hold with high probability with respect to the random choice of
bases. The lower bounds we obtain are optimal up to additive constants and
allow us to establish the conjecture by Wehner and Winter on the asymptotic
behavior of constants in entropic uncertainty relations
as the dimension tends to infinity.

As a tool we develop estimates on the maximum operator norm of a submatrix of a
fixed size of a random unitary matrix distributed according to the Haar measure,
which are of an independent interest.
\end{abstract}

\maketitle


\section{Introduction}
Uncertainty relations belong to key features of quantum theory. In the original
approach of Heisenberg \cite{He27}, Kennard~\cite{kennard1927quantenmechanik}
and Robertson \cite{Ro29} one considers the product of variances
characterizing measurements of two non-commuting observables. In a later
complementary approach one  studies entropies of probability vectors associated
with both measurements and derives lower bounds for the sum of the two entropies
\cite{BBM75}.

State independent bounds for any two orthogonal measurements
performed on a state from a Hilbert space ${\cal H}_N$
of a finite dimension $N$ were obtained first by Deutsch \cite{De83}
and later improved by Maassen and Uffink \cite{MU88}.
The problem is entirely specified by the unitary matrix $U$
defining the transition from one measurement basis to the other one. The bounds
of \cite{De83} and \cite{MU88} are both expressed in terms of the absolute value
of the largest entry of $U$.
More information on entropic uncertainty relations
can be found in review articles \cite{WW10,IBBLR}, while
some of their numerous applications
in the theory of quantum information are discussed in
\cite{Berta,Ra12,ZBP13,G++13}.
Certain improvements with respect to the result of Maassen and Uffink
have been recently obtained in
\cite{FGG13,puchala2013majorization,CP14,rudnicki2014strong,
gour15,prcpz15}.

Although usually one aims to obtain bounds for two measurements
in bases related by a specific unitary matrix $U$,  alternatively
one may benchmark the quality of a given bound by averaging it
over the set of all unitaries with respect to the Haar measure on the unitary group $U(N)$.
Such an approach was advocated in the papers of Hayden et. al \cite{Hayden2004randomizing}
and of Wehner and Winter~\cite{WW10},
in which the authors considered the special case, where
the number $L$ of measurements taken was a function of the dimension $N$ of the
Hilbert space.

Following their approach we analyze entropic uncertainty principles for a fixed
number of measurements in bases related by random unitary matrices (throughout
the article all random unitary matrices we consider are distributed according to
the Haar measure on the unitary group) and provide lower bounds on the sum of
entropies which hold with high probability and differ from the best possible
only by an additive, dimension independent constant.

Our goals and motivation depend on the number of measurements we consider. For
$L = 2$ measurements, many uncertainty relations are known, including the
Maassen--Uffink bound \cite{MU88}, the majorization bounds
\cite{FGG13,puchala2013majorization}, strong majorization relation of
\cite{rudnicki2014strong} or a recent result by Coles and Piani \cite{CP14}.
While in general such inequalities complement each other, it is of interest to
verify how they perform on typical measurements, i.e. on measurements related by
a random unitary matrix. To answer this question we
derive an optimal entropic uncertainty relation for generic measurements -- see
Theorem~\ref{th:asymptotic-strong} in section~\ref{sec:asym-eur}.

A question of entropic uncertainty relations for a large number $L$ of generic
bases in $N$ dimensional Hilbert space, was posted by Wehner and
Winter~\cite{WW10}. In this work we prove  that, with high probability, the
average entropy is bounded from below by $\frac{L-1}{L} \log N - c$, where $c$
is na additive constant independent of $N$ and $L$. 
This allows us to give the affirmative answer to a strong form of a conjecture
by Wehner and Winter \cite{WW10} -- see Theorem~\ref{prop:uncertainty} and
Corollary~\ref{cor:WW-conjecture} in section~\ref{sec:several-measurements}.
Asymptotic uncertainty relations derived in this work improve estimations on the
quality of the information locking protocols recently obtained by Fawzi et al.
\cite{Fawzi2011}.


Our approach is based on the Schur concavity of entropy which together with the
approach proposed in \cite{rudnicki2014strong}
allows us to reduce the problem of finding lower bounds on the sums of Shannon
entropies, to the problem of finding upper bounds on norms of submatrices of a
random unitary matrix. The latter can be then obtained by employing the
concentration of measure phenomenon on the unitary group. We believe that
estimates of maximum norms of a submatrix of fixed size of a random unitary
matrix in high dimensions are of independent interest as similar quantities have
previously appeared in the context of asymptotic geometric analysis and
compressed sensing.

At a technical level it may be noted that the Schur concavity of entropy allows to
reduce the analysis to functions whose Lipschitz constants behave better (as the
dimension increases) then the Lipschitz constant of the entropy itself, thus
allowing us to obtain the right balance between the complexity of approximation
and available tail bounds.

\medskip

This work is organized as follows. In section~\ref{sec:II} we briefly recall the
Maassen--Uffink relations and their improvements. Bounds for the norms of
submatrices of random unitary matrices, also called their truncations
\cite{ZS00}, are presented in Section~\ref{sec:norms}. Asymptotic entropic
uncertainty relations are analyzed in Section~\ref{sec:asym-eur} for the case of
two measurements, while the case of several measurements is discussed in
Section~\ref{sec:several-measurements}. The presentation and discussion of the
results is concluded in Section VI while the proofs of some lemmas are deferred
to the Appendices.

\section{Entropic uncertainty relations}\label{sec:II}

In this section we present entropic uncertainty relations we are going to study
in the asymptotic case. The most important bound for the sum of entropies is due
to Maassen and Uffink~\cite{MU88}.

Consider a normalized  vector $|\psi\rangle$ belonging to
a $N$--dimensional complex Hilbert space ${\cal H}_N$ and a non-degenerate
observable $A$, whose  eigenstates $\ket{a_i}$, $i =1,\ldots,N$, form an
orthonormal basis of ${\cal H}_N$. The probability that this observable measured
in the state $|\psi\rangle$ gives the $i$--th outcome is given by
$p_i^{\psi}=|\<a_i|\psi\>|^2$. Clearly $\sum_{i=1}^Np_i^\psi = 1$, so the vector
$p^\psi = (p_1^\psi,\ldots,p_N^\psi)$ can be identified with a probability
distribution on the set $\{1,\ldots,N\}$. The uncertainty associated to the
measurement $A$ can be then described by the Shannon entropy of $p^\psi$,
defined as
\begin{displaymath}
H(p^\psi)=-\sum_{i=1}^N p^\psi_i\ln p^\psi_i.
\end{displaymath}

Consider now another observable $B$ and let $|b_i\>$, $i = 1,\ldots,N$, be its
eigenstates. Let $ q^\psi$ be the probability distribution associated with $B$,
i.e. $q = (q^\psi_1,\ldots,q^\psi_N)$, where $q_j^{\psi}=|\<b_j|\psi\>|^2$. The
uncertainty corresponding to $B$ can be quantified by the corresponding
Shannon's entropy $H(q^\psi)$. If the observables $A$ and $B$ do not commute,
then the sum of both entropies for any state $|\psi\rangle$ is bounded from
below, and (as one can easily see) the bound depends only on the unitary matrix
$U = (U_{ij})_{i,j=1}^N$, where $U_{ij}=\langle a_i|b_j\rangle$.

In~1988 Maassen and Uffink  \cite{MU88}
obtained the result  of the form
\begin{equation}
\label{MU}
H(p^\psi) + H(q^\psi) \geq - \ln c^2\equiv B_{MU} 
\end{equation}
where $c= \max_{ij} |U_{ij}|$.
The Maassen-Uffink bound has been recently improved in the whole range of the
parameter $c$ by Coles and Piani \cite{CP14} who provided
a state independent bound
\begin{equation}\label{ColPia}
H\left(p^\psi\right)+H\left(q^\psi\right)\geq
-\ln c^2+\left(\frac12 -\frac{c}{2} \right)
\ln\frac{c^2}{c^2_{2}}\equiv B_{\textrm{CP}},
\end{equation}
with $c_{2}$ being the second largest value among $\left|U_{ij}\right|$, $1\le i,j \le N$.
Since $c_{2}\leq c$, the second term in (\ref{ColPia}) is a non-negative
correction to (\ref{MU}).


Let us now pass to uncertainty relations based on the Schur concavity of the
Shannon entropy.
They will take into account not only the largest or the two largest
elements of the transition matrix $U$, but the behavior of the operator norms
of all submatrices of $U$.

Let us first introduce some auxiliary notation related to matrices. By $U(N)$ we
will denote the unitary group of $N\times N$ unitary matrices. For $U =
(U_{ij})_{i,j=1}^N \in U(N)$ and nonempty sets $I,J \subset \{1,\ldots,N\}$, let
$U(I,J) = (U_{ij})_{i\in I, j \in J}$, i.e. $U(I,J)$ is the matrix obtained from
$U$ by restricting to rows and columns corresponding to the elements of $I$ and
$J$ respectively. For a matrix $M$, by $\|M\|$ we denote its operator norm,
equal to its largest singular value, $\sigma_{\rm max}(M)$. Finally, for $1\le
n,m \le N$ we define
\begin{equation}\label{eq:definition-U-n-m}
\| \widehat{U}^{(n,m)}\| = \max\Big\{ \| U(I,J)\| \colon I,J \subset \{1,\ldots,N\},
|I| = n, |J| = m\Big\},
\end{equation}
i.e. $\| \widehat{U}^{(n,m)}\|$ is the maximal norm of a submatrix of $U$ of
size $n \times m$.

For any fixed matrix $U$ we shall introduce a set of $N$ coefficients
\begin{equation}\label{eqn:def-sk}
s_k  := \ \max \left\{  || \widehat{U}^{(1,k)}||,\ ||\widehat{U}^{(2,k-1)}||, \dots ,
 \  ||\widehat{U}^{(k,1)}|| \right\},\; k =1,\ldots,N.
\end{equation}

In the next step we define coefficients
\begin{equation}\label{eq:definition_of_R}
R_k = \left(\frac{1 + s_k}{2} \right)^2, \ k=1, \dots, N,
\end{equation}
so that $\left( \frac{1+c}{2}\right)^2 = R_1 \leq R_2 \leq \dots \leq R_N=1$.

\medskip

Recall also that if $x, y \in \R^N$ have nonnegative coordinates then we say
that $x$ is majorized by $y$ (which we denote by $x \prec y$) if
for $k=1,\ldots, N$, $\sum_{i=1}^k x^{\downarrow}_i \le \sum_{i=1}^k
y^{\downarrow}_i$ and $\sum_{i=1}^N x^{\downarrow}_i = \sum_{i=1}^N
y^{\downarrow}_i$, where  $x^{\downarrow}_i$ is the non-increasing rearrangement
of the numbers $x_i$. We say that a function $f \colon \R_+^N \to \R$ is Schur
concave if $f(x) \ge f(y)$ whenever $x \prec y$. It is well known that the
function $f(x) = \sum_{i=1}^N - x_i \ln x_i$ is Schur concave (see e.g.
\cite{BZ2006}).

\medskip
We are now ready formulate a result proved in \cite{puchala2013majorization}.

\begin{theorem}\label{th:main-theorem_P2013}
Let $(|a_i\>)_{i=1}^N$ and $(|b_i\>)_{i=1}^N$ be two orthonormal bases in ${\cal
H}_N$ and $U = (\< a_i|b_j\>)_{i,j=1}^N$ be the corresponding transition matrix.
Define
$Q = \left(R_1,R_2-R_1,R_3-R_2, \dots, R_N-R_{N-1} \right),$
where the coefficients $R_i$ are given above.
Then for any pure state $|\psi\> \in {\cal H}_N$, the probability vectors
$p^\psi = (p^\psi_1,\ldots,p^\psi_N)$, $q^\psi = (q^\psi_1,\ldots,q^\psi_N)$
with $p_i^{\psi}=|\<a_i|\psi\>|^2$, $q_i^{\psi}=|\<b_i|\psi\>|^2$, satisfy
$p^\psi\otimes q^\psi \prec Q.$
\end{theorem}

Notice, that from the above theorem and the Schur concavity of the Shannon's
entropy we obtain directly the following corollary.

\begin{corollary}\label{cor:renyi-bounds}
In the setting of Theorem \ref{th:main-theorem_P2013},
\begin{equation}\label{eqn:prod-bound}
\begin{split}
\min_{\psi} \Big( H (p^\psi) + H (q^\psi)\Big)\; \geq \; H \left(Q\right),
\end{split}
\end{equation}
where the minimum is taken over the set of all pure states $|\psi\> \in {\cal H}_N$.

\end{corollary}


Recently  an improved version of majorization entropic uncertainty relations was
derived  in ~\cite{rudnicki2014strong}.

\begin{theorem}\label{th:strong}
In the setting of Theorem \ref{th:main-theorem_P2013},
we define the numbers $x_i, i = 1,\ldots,2N$ by the equality $p^\psi\oplus
q^\psi = (x_1,\ldots,x_{2N})$. Then for $k = 1,\ldots,2N$,
\begin{equation}\label{eq:partial-sum-estimate}
\sum_{i=1}^k x^{\downarrow}_i \le 1 + s_{k-1},
\end{equation}
where we additionally set $s_0 = 0$. As a consequence,
\begin{equation}
p^\psi \oplus q^\psi \prec (1,s_1,s_2-s_1,s_3-s_2,\dots,s_N-s_{N-1}).
\end{equation}

\end{theorem}
The majorization relation of Theorem \ref{th:strong} implies the following
uncertainty relation
\begin{equation}
\label{eqn:strong-relation}
\min_{\psi} \Big(H(p) + H(q)\Big) \geq H((s_1,s_2-s_1,s_3-s_2,\dots,s_N-s_{N-1})),
\end{equation}
where as usual the minimum is taken over the set of all pure states.

Note that in this case we apply majorization techniques working with positive
vectors which are not normalized to unity. In paper~\cite{rudnicki2014strong} it
has been shown, that the bound~\eqref{eqn:strong-relation} based on the direct
sum is not weaker than the bound~\eqref{eqn:prod-bound} based on the tensor
product of probability vectors.

Another result proved in \cite{rudnicki2014strong} is an uncertainty relation
for many measurements, which we now recall. For $L\ge 2$ consider $N\times N$
unitary matrices $U_1,\ldots,U_L$ and let $\ket{u^{(i)}_j}$ be the $j$-th column
of $U_i$. Consider the probability distributions $p^{(i)}$, $i = 1,\ldots,L$
given by $p^{(i)}_j = |\scalar{u^{(i)}_j}{\psi}|^2$, $i=1,\ldots,L$, $j =
1,\ldots,N$. Note that to simplify the notation we suppress here the dependence
on $|\psi\>$.

Let finally $U$ be the concatenation of matrices $U_1,\ldots,U_L$ and for a set
$I \subset \{1,\ldots,LN\}$ with $|I| = k$, let $U_I$ be the $N \times k$ matrix
obtained from $U$ by selecting the columns of $U$ corresponding to the set $I$.
Define for $k=0,\ldots,NL-1$,
\begin{equation}\label{eq:definition-S_k}
\mathcal{S}_k = \max\{\|U_I\|^2\colon I \subset \{1,\ldots,LN\}, |I| = k+1\}.
\end{equation}
Note that $\mathcal{S}_0 = 1$, independently of the choice of unitary matrices $U_i$.

The following theorem was proved in \cite{rudnicki2014strong}.

\begin{theorem}\label{th:strong-many_measurements}

In the setting described above, define the coefficients $x_1,\ldots,x_{NL}$ by
the equality
$p^{(1)}\oplus\cdots\oplus p^{(L)} = (x_1,\ldots,x_{NL}).$
Then for $k \le NL$,
$\sum_{i=1}^k x^\downarrow_i\le \mathcal{S}_{k-1}.$
As a consequence,
$
p^{(1)}\oplus\cdots\oplus p^{(L)} \prec (\mathcal{S}_0,\mathcal{S}_1-\mathcal{S}_0,\mathcal{S}_2-\mathcal{S}_1,\ldots,\mathcal{S}_{LN} - \mathcal{S}_{LN-1})
$
and
\begin{equation}
\sum_{i=1}^L H(p^{(i)}) \ge - \sum_{i=1}^{LN} (\mathcal{S}_{i} - \mathcal{S}_{i-1})\ln(\mathcal{S}_{i} - \mathcal{S}_{i-1}).
\end{equation}
\end{theorem}
We note that for general $L > 2$, it is an open problem  to construct
deterministic unitary $N\times N$ matrices $U_1,\ldots,U_L$ such that for all
pure states $|\psi\>$,
\begin{equation}\label{eq:strong-relation-several-measurements}
\sum_{i=1}^L H(p^{(i)}) \ge (L-1)\ln N - CL,
\end{equation}
where $C$ is a constant independent of $N$.  A deterministic construction is
known only for $L = N + 1$, in which case it was proved by Ivanovic
\cite{Ivanovic} and S{\'a}nchez-Ruiz \cite{Sanchez-Ruiz} that the above bound
holds for unitary matrices corresponding to a maximal set of mutually unbiased
basis. On the other hand, as shown in \cite{Ballester2007}, if $N$ is an even
power of a prime number and $L \le \sqrt{N}+1$, then there exist $L$ mutually
unbiased bases such that for some state $\ket{\psi}$,
\begin{equation}
\sum_{i=1}^L H(p^{(i)}) =  \frac{1}{2}L \ln N.
\end{equation}
In particular this shows that the approach of \cite{Ivanovic,Sanchez-Ruiz}
cannot be generalized to arbitrary $L$.

To the best of our knowledge, for a `small' number of measurements the only
available constructions of bases satisfying
\eqref{eq:strong-relation-several-measurements} are given by the random choice
of bases and work for $L \ge \ln^4 N$ \cite{Hayden2004randomizing}. We will
discuss them in Section \ref{sec:several-measurements} (together with related
work \cite{Fawzi2011}), where we show that random bases provide strong
uncertainty relations also for a smaller number of measurements.

\section{Norms of truncations of random unitaries}\label{sec:norms}
In this section we will provide estimates for the operator norms of submatrices
of a random unitary matrix, which as seen in the previous section, appear in
majorization entropic uncertainty principles. These estimates will become
crucial in the proofs of entropic uncertainty principles for random unitaries.
We emphasize that although from the point of view of uncertainty principles,
bounds on norms of submatrices are simply a tool, we have decided to state them
in a separate section as we believe that they may be of independent interest,
especially from the perspective of Random Matrix Theory or Asymptotic Geometric
Analysis.

Before stating our results let us recall some basic notions related to random
unitary matrices. As is well known, the unitary group $U(N)$ of all $N\times N$
unitary matrices admits a unique probability measure invariant under left and
right multiplications, i.e. the Haar measure. In what follows by a $N\times N$
random unitary matrix we will always mean a random element of the group $U(N)$
distributed according to the Haar measure. Usually we will denote such a random
matrix by $U$, suppressing the dependence on $N$, as it is customary in the
Random Matrix Theory literature.

\medskip

Motivated by the result of Maassen and Uffink, for $U = (U_{ij})_{i,j=1}^N$, we
denote
\begin{equation}
c(U) = \max_{1\le i,j\le N} |U_{ij}|.
\end{equation}

The behavior of $c(U)$ for random unitaries was studied by Jiang \cite{Jiang},
who obtained

\begin{theorem}\label{th:jiang}
If $U$ is a $N\times N$ random unitary matrix, then for all $\varepsilon > 0$,
\begin{equation}\label{eqn:s1-bound}
\p\left(
(1 - \varepsilon)\sqrt{\frac{2}{N} \ln N}
\leq c(U) \leq
(1 + \varepsilon)\sqrt{\frac{2}{N} \ln N}
\right) \to 1
\text{ as } N\to \infty.
\end{equation}
\end{theorem}

The next theorem is a generalization of the result obtained by Jiang to the
maximal norm of submatrices of a random unitary matrix (as defined by
\eqref{eq:definition-U-n-m}).

\begin{theorem}\label{th:submatrices-fixed-size}
For any fixed positive integers $n,m$ and any $\varepsilon > 0$, if $U$ is a
$N\times N$ random unitary matrix, then
\begin{equation}\label{eq:submatrices-fixed-size}
\p\Big((1-\varepsilon)\sqrt{\frac{n+m}{N}\ln N } \le \|\widehat{U}^{(n,m)}\| \le (1+\varepsilon)\sqrt{\frac{n+m}{N}\ln N } \Big) \to 1 \text{ as } N\to \infty.
\end{equation}
\end{theorem}

The above theorem works for fixed $n,m$, independent of the dimension $N$. Its
proof is based on the following result, which provides an estimate on the
maximal norm $\|\widehat{U}^{(n,m)}\|$ for arbitrary $n,m \le N$. Before we
formulate the theorem, let us recall also that $o(1)$ denotes any sequence which
converges to zero as $N \to \infty$, in particular for non-vanishing sequences
of real numbers $a_N$ and $b_N$, we have $a_N = (1+o(1))b_N$ if and only if
$\lim_{N\to \infty} \frac{a_N}{b_N} = 1$.

\begin{theorem}\label{th:upper}
Let $U$ be a $N\times N$ random unitary matrix. Then
\begin{equation}
\label{conc_Ank}
\p \left( \left| \|\widehat{U}^{(n,m)}\| -
\Ex  \|\widehat{U}^{(n,m)}\| \right|\geq t \right)\leq
2\exp\Big(-\frac{Nt^2}{12}\Big)\quad \mbox{ for }t\geq 0.
\end{equation}
Moreover, for any $0<\ve<1/3$,
\begin{equation}\label{eq:norm_expectation_estimate}
\Ex \|\widehat{U}^{(n,m)}\| \leq
\frac{1}{1-2\ve-\ve^2}\sqrt{\frac{2}{2N-1}}\left(m\ln \frac{eN}{m}+n\ln\frac{eN}{n}+2(n+m)\ln(1+\frac{2}{\ve})\right)^{1/2}.
\end{equation}
In particular for any fixed $n,m$ and $N\rightarrow\infty$,
\begin{equation}
\Ex \|\widehat{U}^{(n,m)}\| \leq (1+o(1))\sqrt{\frac{m+n}{N}\ln N}.
\label{trun1}
\end{equation}
\end{theorem}

In the special case, when one of the parameters $n,m$ equals to one, more
precise estimates are provided by subsequent theorems. The first one relies on a
geometric argument, exploiting the fact that in the special situation when $n =
1$ or $m=1$, the norms we consider are Euclidean.

\begin{theorem}\label{th:single-column}
If $U$ is a $N\times N$ random unitary matrix, then for all $\varepsilon >0$,
\begin{equation}\label{eq:harmonic}
\min_{1\le n\le N}\p \left(
\frac{n}{N}(1 + H_N - H_n) - \varepsilon \leq \|\widehat{U}^{(n,1)}\|^2 \leq \frac{n}{N}(1 + H_N - H_n) + \varepsilon
\right) \to 1 \text{ as } N \to \infty,
\end{equation}
where $H_m =  \sum_{j=1}^m 1/j$ denotes the $m$-th harmonic number.
\end{theorem}

The next theorem provides a complete characterization of the behavior of
$\|\widehat{U}^{(n,1)}\|$ for large random unitary matrices. Its proof relies on
a combination of Theorem \ref{th:submatrices-fixed-size}, which allows to handle
the case of `small' $n$ and Theorem \ref{th:single-column} which provides good
estimates for large values of $n$.
\begin{theorem}\label{th:one-times-n}
Let $U$ be a $N\times N$ random unitary matrix. For all $\varepsilon > 0$,
\begin{equation}\label{eq:uniform-one-times-n}
\p\bigg( \forall_{1 \le n \le N} \; (1-\varepsilon)\sqrt{\frac{n+1}{N}\Big(1 + \ln\Big(\frac{N}{n}\Big)\Big)} \le \|\widehat{U}^{(n,1)}\| \le (1+\varepsilon)\sqrt{\frac{n+1}{N}\Big(1 + \ln\Big(\frac{N}{n}\Big)\Big)}\bigg) \to 1
\end{equation}
as $N \to \infty$.
\end{theorem}

Observe, that Theorem \ref{th:one-times-n} provides a complete description of
asymptotic behavior of the whole sequence $\|\hat{U}^{(n,1)}\|$, $n =
1,2,\dots,N-1$, while by setting $m = 1$ in Theorem~\ref{th:upper} one obtains
non-trivial bounds on the norm only for $n \leq N/ \ln N -1$.

Proofs of all the results described in this section are deferred to Appendix \ref{submatr2}.

\section{Asymptotic entropic uncertainty relations} \label{sec:asym-eur}

In this section we assume that $N \gg 1 $ and analyze the asymptotic behavior of
entropic uncertainty relations for random unitary matrices. We consider two
orthogonal von Neumann measurements with respect to two bases related by a
random unitary matrix $U$ distributed according to the Haar measure on the
unitary group. We note that as mentioned in Section \ref{sec:II}, if
$(|a_i\>)_{i=1}^N$, $(|b_i\>)_{i=1}^N$ are two orthonormal bases in ${\cal H}_N$
and for a pure state $|\psi\>$ in ${\cal H}_N$, the vectors $p^\psi =
(p_1^\psi,\ldots,p_N^\psi)$, $q^\psi = (q_1^\psi,\ldots,q_N^\psi)$ are given by
\begin{displaymath}
p_i^{\psi}=|\<a_i|\psi\>|^2, \; q_i^{\psi}=|\<b_i|\psi\>|^2,
\end{displaymath}
then the quantity
\begin{displaymath}
\min_{\psi} \Big(H(p^\psi) + H(q^\psi)\Big)
\end{displaymath}
depends only on the unitary transition matrix $U = (U_{ij})_{i,j=1}^N$, given by
$U_{ij} = \<a_i|b_j\>$. Therefore, when $U$ is a $N\times N$ random unitary
matrix, we can speak about probabilities of the form
\begin{displaymath}
\p\Big(\min_{\psi} \Big(H(p^\psi) + H(q^\psi)\Big) \ge r\Big).
\end{displaymath}
There is clearly a slight abuse of notation in this convention since to define
$p^\psi$ or $q^\psi$ one has to choose the bases $(|a_i\>)_{i=1}^N$,
$(|b_i\>)_{i=1}^N$, but it should not lead to ambiguity. Alternatively, to give
definite meaning to $p^\psi$ and $q^\psi$ one can decide (without loss of
generality) that ${\cal H}_N = \C^N$, $(|a_i\>) _{i=1}^N$ is some fixed basis of
${\cal H}_N$ (e.g. the standard one) and $|b_i\> = U|a_i\>$.

Let us start our study of uniform uncertainty principles in the random setting
by evaluating the typical behavior of deterministic bounds of Section
\ref{sec:II}. We emphasize that this part of our analysis will follow easily
from known bounds on maximal entries of random unitary matrices. The more
challenging part will be to obtain optimal bounds, given in Theorem
\ref{th:asymptotic-strong} below, which will allow us to conclude that in
generic situations bounds of Maassen-Uffink type give only sub-optimal results.

The first proposition evaluates the performance of the Maassen-Uffink entropic
uncertainty relation.

\begin{proposition}\label{th:asymptotic-MU}
Let $U$ be a $N\times N$ random unitary matrix and let $B_{MU} = - \ln c^2(U)$.
Then for any $\varepsilon > 0$,
\begin{equation}
\p(\ln N - \ln \ln N - \ln 2 - \varepsilon \le B_{MU} \le \ln N - \ln \ln N -
\ln 2 + \varepsilon) \to 1
\end{equation}
as $N \to \infty$.
\end{proposition}

The interpretation of this result in the context of entropic uncertainty
relations is that in sufficiently large dimension $N$ the lower bound obtained
by an application of the Maassen-Uffink inequality to a typical (i.e. related by
a random unitary matrix) pair of orthogonal von Neumann measurements is (with
probability close to one)
\begin{equation}
\min_{\psi} \Big(H(p^\psi) + H(q^\psi)\Big) \geq \ln N - \ln \ln N - \ln 2 - o(1),
\end{equation}
As we will see in Theorem \ref{th:asymptotic-strong}, this bound is off by the
term of the order $\ln\ln N$. This shows that while in the extreme situation
(e.g. when the measurements are related by a Hadamard matrix), the
Maassen-Uffink bound cannot be improved, its typical performance is sub-optimal.

We postpone the proof of the above proposition till Appendix
\ref{sect:simple-maj}. Here we just mention that the argument is an elementary
corollary to Jiang's estimates on $c(U)$, given in Theorem \ref{th:jiang}.

\medskip

In view of the discussion above, one may wonder whether in typical situations it
is possible to obtain a significant gain by employing the Coles and Piani
relation \eqref{ColPia} instead of the Maassen-Uffink bound. It turns out
however that this will not provide a notable improvement, since for large $N$
with high probability we have $c(U) \simeq c_2(U)$ (recall that $c_2(U)$ is the
second largest number among $|U_{ij}|$, $1\le i,j\le N$). This is formalized in
the following proposition.

\begin{proposition}\label{prop:CP-asymptotic}
Let $U$ be a random $N\times N$ unitary matrix and let
\begin{equation}
B_{\textrm{CP}}(U) = -\ln c^2(U)
+
\left( 1 - c(U)\right) \ln\frac{c(U)}{c_2(U)},
\end{equation}
where $c(U)$ and $c_2(U)$ denote respectively the largest and second largest
absolute value of an entry of $U$.
Then for every $\varepsilon > 0$,
\begin{equation}\label{eqn:CP-asymptotic}
\p( B_{\textrm{CP}} \leq \ln N - \ln \ln N - \frac12 \ln 2 + \varepsilon) \to 1
\end{equation}
as $N \to \infty$.
\end{proposition}

As one can see from the above proposition, in typical situations the gain
obtained from the Coles and Piani relation with respect to the Massen-Uffink
bound is just $2^{-1}\ln 2$. Proposition \ref{prop:CP-asymptotic} will be proved
in Appendix \ref{sect:simple-maj}. Let us remark that it will be again a
relatively simple corollary to the estimates on $c(U)$ and $\|\hat{U}^{(1,2)}\|$
given in Theorems \ref{th:jiang} and \ref{th:submatrices-fixed-size}.

\medskip

Let us now turn to the majorization entropic uncertainty relation discussed in
Section \ref{sec:II}. As shown
in~\cite{puchala2013majorization} in many cases it provides a tighter bound than
the Maassen-Uffink relation (\ref{MU}), however it turns out that this is not
the case for typical measurements in high dimension as we have the following
proposition.

\begin{proposition}\label{prop:majorization-weak}
Assume that $N \ge 4$ and let $U$ be any $N\times N$ unitary matrix and let
\begin{equation}\label{eqn:def-Q}
Q = \left(R_1,R_2-R_1,R_3-R_2, \dots, R_N-R_{N-1} \right),
\end{equation}
where the coefficients $R_i$ are given by formula \eqref{eq:definition_of_R}.
Then
\begin{equation}\label{eq:majo-weak}
H(Q) \leq \frac{3}{4} \ln(N-1) + H\left(\frac{1}{4},\frac{3}{4}\right) = \frac{3}{4}\ln(N-1) + \frac{1}{4}\ln 4 + \frac{3}{4}\ln\frac{4}{3}.
\end{equation}
\end{proposition}

Let us note that in the above proposition the matrix $U$ is not random, it can
be any $N\times N$ unitary matrix. Together with examples presented in
\cite{puchala2013majorization}, the inequality \eqref{eq:majo-weak} shows that
entropy estimates based on tensor product majorization of Theorem
\ref{th:main-theorem_P2013} do not perform well in typical or extremal
situations, even though they can still outperform the classical Maassen-Uffink
bound when the entropy is small. This is intuitively clear, since the
probability distribution $Q$ has an atom $R_1$ of size at least $\frac{1}{4}$.
The proof Proposition of \ref{prop:majorization-weak} is presented in
Appendix~\ref{sect:simple-maj}.

\bigskip
We will now pass to the first main result of the article, i.e. to optimal (up to
a universal additive constant) entropic uncertainty principles for typical
measurements, which hold with high probability on the unitary group. We
emphasize that the method of proof will rely heavily on strong (direct sum)
majorization of Theorem \ref{th:strong}, more specifically on the bound
\eqref{eq:partial-sum-estimate}, combined with the results of Section
\ref{sec:norms}. Our result shows in particular that strong majorization
techniques of \cite{rudnicki2014strong} perform in typical high-dimensional
scenarios in an almost optimal way. We refer the reader to
\cite{rudnicki2014strong} for a comparison of the inequality of Theorem
\ref{th:strong} with the result~\eqref{eqn:prod-bound} and with the
Maassen-Uffink bound in deterministic, low-dimensional situations and here we
just mention that Theorem \ref{th:strong} is stronger than~\eqref{eqn:prod-bound} and in general incomparable with the relation of Maassen
and Uffink (one can construct examples in which any of the bounds outperforms
the other one).

The following theorem,
provides optimal uncertainty relations for generic measurements.

\begin{theorem}\label{th:asymptotic-strong}
Let $U$ be a $N\times N$ random unitary matrix and let $\constantCa =  3.49$.
Then
\begin{equation}
\p \Big(\min_{\psi} (H(p^\psi) + H(q^\psi) )\geq \ln N - \constantCa \Big) \to 1
\end{equation}
as $N \to \infty$.
\end{theorem}
Recall the definition of the parameters $s_k$ given in \eqref{eqn:def-sk}. The
proof of Theorem \ref{th:asymptotic-strong} will be based on the following
proposition proved in Appendix~\ref{sec:uniform-bound-for-sk},

\begin{proposition}\label{prop:uniform-bound-for-sk}
With probability tending to 1 as $N \to \infty$,
\begin{align}\label{eq:basic}
s_k \le \sqrt{\constantCf\frac{k+1}{N} \left(1+\ln\left(\frac{2 N}{k+1}\right)
\right)}
\quad \textrm{for}\; 1\leq k\leq N,
\end{align}
where $\constantCf = 4.18$.
\end{proposition}

Now we are in position to prove the Theorem~\ref{th:asymptotic-strong},
\begin{proof}[Proof of Theorem \ref{th:asymptotic-strong}]
Let as fix a unitary vector $\ket{\psi}$ and let $p:=p^\psi$, $q:=q^\psi$.
Recall that $\constantCf = 4.18$. We define the sequence $m_i$ as
\begin{equation}
m_1 = \sqrt{\constantCf\frac{2}{N}( 1 + \ln(N))},
\end{equation}
and for $2 \leq i \leq 2N -1$, 
\begin{equation}
m_i = \sqrt{\constantCf\frac{i+1}{N}\left( 1 + \ln\left(\frac{2 N }{i+1}\right)
\right)}
-
\sqrt{\constantCf\frac{i}{N} \left( 1 + \ln \left( \frac{ 2 N }{ i } \right)
\right)} > 0,
\end{equation}
which 
we can rewrite as
\begin{equation}
m_i = \sqrt{ 2\constantCf } \left( f \left(\frac{i+1}{2N} \right) - f
\left(\frac{i}{2N}
\right) \right),
\end{equation}
where $f \colon (0,e) \to \R$ is given by $f(x) = \sqrt{x\left( 1 - \ln x
\right)}$. The function $f$ is concave, which can be verified by simple
calculations, i.e.
\begin{equation}
\frac{d}{dx} f(x) = \frac{- \ln x}{2 \sqrt{x\left( 1 - \ln x \right)}}
\end{equation}
and
\begin{equation}
\frac{d^2}{dx^2} f(x) =
\frac{
- (1 - \ln x)^2 -1
}{
4 (x (1 - \ln x ))^{3/2}
} < 0.
\end{equation}
From concavity we obtain that for $2 \le i \le 2N - 1$,
\begin{equation}
m_i \le \sqrt{2\constantCf} \frac{1}{2N} \frac{d}{dx}f \left(\frac{i}{2N}
\right)
=\frac{1}{N} \frac{ \sqrt{\constantCf} \ln ( \frac{2N}{i} ) }{ 2\sqrt{
\frac{i}{N} (\ln(\frac{2N}{i} ) + 1 ) } }.
\end{equation}
Note that
\begin{equation}
m_1 + \sum_{i=2}^{N}
\frac{1}{N}\frac{\sqrt{\constantCf}\ln(\frac{2N}{i})}{2\sqrt{\frac{i}{N}(\ln(\frac{2N}{i})+1)}}
\ge
\sum_{i=1}^N m_i
=
\sqrt{\constantCf\frac{N+1}{N} \left(1 + \ln\left( \frac{2 N}{N+1}\right)
\right)} > 1.
\end{equation}
Let $N_0$ be the greatest integer not exceeding $N$, such that
\begin{equation}
 m_1 + \sum_{i=2}^{N_0}
 \frac{1}{N}\frac{\sqrt{\constantCf}\ln(\frac{2N}{i})}{2\sqrt{\frac{i}{N}(\ln(\frac{2N}{i})+1)}}
  \le 1
\end{equation}
and define a vector $r = (r_1,\ldots,r_{N_0+1}) \in \R^{N_0+1}$ by specifying
its coordinates as follows. Set $r_1 = m_1$ and
\begin{equation}
r_i =
\frac{1}{N}\frac{\sqrt{\constantCf}\ln(\frac{2N}{i})}{2\sqrt{\frac{i}{N}(\ln(\frac{2N}{i})+1)}}
\end{equation}
for $i = 2,\ldots,N_0$. As the last coordinate set $r_{N_0+1} = 1 - \sum_{i=1}^{N_0}
r_i$, so $r$ is a probability vector. 
Note that
\begin{align}\label{eq:last_r}
r_{N_0+1} \le
\frac{1}{N}\frac{\sqrt{\constantCf}\ln(\frac{2N}{N_0+1})}{2\sqrt{\frac{N_0+1}{N}(\ln(\frac{2N}{N_0+1})+1)}}
\le \sqrt{\frac{\constantCf\ln N}{N}}.
\end{align}
Let $z$ be the non-increasing rearrangement of $p\oplus q$. For $k \le N_0+1$,
Theorem \ref{th:strong} and Proposition \ref{prop:uniform-bound-for-sk} give
\begin{equation}
\begin{split}
z_1 + \ldots + z_k &\le 1 + s_{k-1} \le 1 +
\sqrt{\constantCf\frac{k}{N}\Big( \ln\Big(\frac{2N}{k}\Big)+1\Big)} \\
&= 1 + m_1+\ldots+m_{k-1}
\le 1 +r_1 +\ldots+r_{k-1}.
\end{split}
\end{equation}
Obviously we also have $z_1+\ldots+ z_{k} \le 2 = 1 +r_1+\ldots+r_{N_0+1}$ for
$k > N_0+1$, and so $z \prec 1\oplus r$.
As a consequence,
\begin{equation}
H(p) + H(q) \ge H(r).
\end{equation}
We will now bound from below the entropy of the vector $r$. We have
\begin{equation}
\begin{split}
H(r) &\ge - \sum_{i=2}^{N_0} r_i\ln r_i
=
- \sum_{i=2}^{N_0}  r_i \ln
\left(\frac{1}{N}\frac{\sqrt{\constantCf}\ln(\frac{2N}{i})}{2\sqrt{\frac{i}{N}(\ln(\frac{2N}{i})+1)}}
\right)\\
&=
\sum_{i=2}^{N_0} r_i \ln N - \sum_{i=2}^{N_0}
\frac{1}{N}\frac{\sqrt{\constantCf}\ln(\frac{2N}{i})}{2\sqrt{\frac{i}{N}(\ln(\frac{2N}{i})+1)}}\ln\left(\frac{\sqrt{\constantCf}\ln(\frac{2N}{i})}{2\sqrt{\frac{i}{N}(\ln(\frac{2N}{i})+1)}}\right)\\
& = \ln N - (m_1 + r_{N_0+1})\ln N - A_N = \ln N -
O\Big(\frac{\ln^{3/2}N}{\sqrt{N}}\Big) - A_N,
\end{split}
\end{equation}
where
\begin{equation}
A_N = \sum_{i=2}^{N_0}
\frac{1}{N}\frac{\sqrt{\constantCf}\ln(\frac{2N}{i})}{2\sqrt{\frac{i}{N}(\ln(\frac{2N}{i})+1)}}
\ln
\left(\frac{\sqrt{\constantCf}\ln(\frac{2N}{i})}{2\sqrt{\frac{i}{N}(\ln(\frac{2N}{i})+1)}}
 \right).
\end{equation}
Above we used \eqref{eq:last_r} and the estimate $m_1 = O(\sqrt{\frac{\ln N}{N}})$.

Let us now bound $N_0/N$ from above. Since
\begin{equation}
m_1 + \sum_{i=2}^{k}
\frac{1}{N}\frac{\sqrt{\constantCf}\ln(\frac{2N}{i})}{2\sqrt{\frac{i}{N}(\ln(\frac{2N}{i})+1)}}
 \ge \sum_{i=1}^k m_i
=
\sqrt{\constantCf\frac{k+1}{N}\left( 1 + \ln\Big(\frac{2 N}{k+1}\Big) \right)},
\end{equation}
we have
$N_0 \le N_1$, where $N_1$ is the largest integer smaller than $N$, such that
\begin{equation}
\sqrt{\constantCf\frac{N_1+1}{N}  \left( 1 + \ln\Big(\frac{2 N}{N_1+1}\Big)
\right)} \le 1.
\end{equation}
We have $N_1/N \to x^\ast$, where $x^\ast$ is the unique solution of
\begin{equation}
\constantCf x^\ast (1 + \ln(2/x^\ast)) = 1.
\end{equation}
Since $\constantCf = 4.18$ we can evaluate numerically that $x^\ast \simeq
0.051$ and so we can write
\begin{equation}
\limsup {A_N} <
\int_0^{0.052}
\frac{\sqrt{\constantCf}\ln\left(\frac{2}{x}\right)}{2
\sqrt{x\ln\left(\frac{2e}{x}\right)}}
\ln\left(\frac{\sqrt{\constantCf}\ln\left(\frac{2}{x}\right)}{2\sqrt{x
\ln\left(\frac{2e}{x}\right)}}\right) dx
\simeq 3.488,
\end{equation}
which ends the proof (note that the integrand above is positive on the interval of integration).
\end{proof}

The state independent lower bound $\ln N - \constantCa$ on the sum of entropies
is
clearly stronger than all the bounds derived from known entropic uncertainty
relations that we have analyzed above. Also it differs from the best possible
one by at most $\constantCa$, since by choosing $\psi$ to be a member of one of
the
bases related to measurements we can enforce the equality $H(p^\psi) = 0$,
whereas trivially $H(q^\psi) \le \ln N$. In fact by taking the randomness into
account one can show that the gap between the result of Theorem
\ref{th:asymptotic-strong} and the optimal one is even smaller, since for random
$U$ and fixed $\psi$, the quantity $H(q^\psi)$ can be interpreted as the entropy
of a random state. An estimation for the mean entropy of a random state  follows
from the work \cite{jones1990entropy} by Jones.  Let $\ket{\psi}$ and
$\ket{\phi}$ be $N$-dimensional normalized vectors in $\C^N$ and  $ \mathrm{d}
\Omega_{\phi} $ be the unique, normalized unitary invariant measure
$d\Omega_{\phi}$ upon the set of pure quantum states. Jones analyzed the mean
value of the following entropy
\begin{equation}
H(1,1) = - N \int |\scalar{\psi}{\phi}|^2 \ln(|\scalar{\psi}{\phi}|^2)
\mathrm{d} \Omega_{\phi},
\end{equation}
and derived its asymptotic behavior
\begin{equation}
H(1,1) = \Psi(N + 1) - \Psi(2) \sim_{N \to \infty} \ln(N) - \Psi(2) + o(1).
\end{equation}

Here $\Psi(z)=\frac{\Gamma'(z)}{\Gamma(z)}$ denotes the digamma function, and
$\Psi(2) = 1 - \gamma \simeq 0.42$, where $\Gamma(z) = \int_0^\infty
x^{z-1}e^{-x}dx$ is the Gamma function and $\gamma$ is the Euler constant. Note,
that $H(1,1)$ is also the mean value of the entropy of a probability vector $p_i
= |\bra{\psi}U\ket{i}|^2,\  i=1,2\dots,N$ describing von~Neumann measurement of
a fixed pure state $\ket{\psi}$ with respect to a basis related to a random
unitary matrix $U$ or equivalently entropy of a pure random state with respect
to a fixed basis. Since it is known that the Shannon's entropy of a pure random
state concentrates strongly around the expectation (see Appendix B.2. of
\cite{Hayden2004randomizing}), by combining the above result with
Theorem~\ref{th:asymptotic-strong} we arrive at a sandwich relation described by
the following theorem.

\begin{theorem}
Let $U$ be a $N\times N$ random unitary matrix.  Let $\constantCzero$ be any
real number
smaller than $1 - \gamma \simeq 0.42$ and let $\constantCa =  3.49$. Then
\begin{equation}
\p( \ln N - \constantCzero  \geq   \min_{\psi} (H(p^\psi) + H(q^\psi))   \geq
\ln N -
\constantCa) \to 1  ,
\label{sandw}
\end{equation}
as $N \to \infty$.
\end{theorem}

\section{Several measurements} \label{sec:several-measurements}

Here we will consider
the case of an arbitrary number of $L$ orthogonal measurements.
Assume that the measurement bases are determined by independent random
unitary matrices $U_1,\ldots,U_L$ of size $N$.
For a state $\ket{\psi}$, let $p^{(\psi,i)} =
(p^{(i)}_1,\ldots,p^{(i)}_N)$, with $p^{(i)}_j = |\langle
u^{(i)}_j|\psi\rangle|^2$, where $u^{(i)}_j$ is the $j$-th column of $U_i$.

Uncertainty relations for random unitaries were studied in
\cite{Hayden2004randomizing,Fawzi2011}. In the special case as the number $L$ of
measurements grows with the dimension $N$ as $  \constantCb  \ln^4 N$
the following asymptotic bound for the average entropy was
derived~\cite{Hayden2004randomizing}
\begin{equation} \label{eqn:Hayden2004randomizing}
\p\Big( \min_{\psi}\frac{1}{L}\sum_{i=1}^L H(p^{(\psi,i)}) \ge \ln N - \constantCb\Big)
\to 1.
\end{equation}



In our work we improve the above result and relax the assumption that number of measurements $L$ and the dimensionality of the system $N$ are related.
The second main result of this paper shows that uniform unitaries satisfy strong uncertainty
relations for an arbitrary number of measurements.

\begin{theorem}\label{prop:uncertainty}
There exists a universal constant $\constantCc$ such that if
$U_1,\ldots,U_L$ are
independent $N\times N$ random unitary matrices, then
\begin{equation}
\p\Big(\min_{\psi}\frac{1}{L}\sum_{i=1}^L H(p^{(\psi,i)}) \ge \frac{L-1}{L}\ln
N - \constantCc\Big) \to 1
\label{multi}
\end{equation}
as $N \to \infty$. Moreover, the convergence is uniform in $L \ge 2$.
\end{theorem}

%


Note that for $L \gg \ln N$ we have $\frac{L-1}{L} \ln N = \ln N - o(1)$, so in
particular our result recovers~\eqref{eqn:Hayden2004randomizing}.

Before providing the proof we will present a few
comments concerning our approach and emphasize the differences with arguments
in \cite{Hayden2004randomizing} or \cite{Fawzi2011}.
We rely on strong majorization relations obtained in \cite{rudnicki2014strong},
 which we combine with
estimates of norms of submatrices of a random unitary matrix presented in
Section \ref{sec:norms}. The main probabilistic ingredient of our proof is the
concentration of measure phenomenon combined with discretization, also
used in \cite{Hayden2004randomizing,Fawzi2011}. The advantage of the
majorization approach stems from the fact that it reduces the problem to the
analysis of norms of matrices, which are $1$-Lipschitz functions of the matrix
(with respect to the Hilbert-Schmidt norms), whereas the Lipschitz constant of
the Shannon's entropy as a function of a state increases with the dimension. A
better Lipschitz constant yields stronger concentration results which gives more
freedom in choosing appropriate approximating sets and as a consequence allows
to find the right balance between the complexity of the problem in dimension $N$
and available probabilistic bounds.

In the proof of Theorem \ref{prop:uncertainty} we will use the following technical lemma, which will be proved
in Appendix~\ref{sec:lemma:norms_multiple}. Recall the definition of $\mathcal{S}_k$ given in formula \eqref{eq:definition-S_k} and the notation used therein:  $U$ is the concatenation of matrices $U_1,\ldots,U_L$ and for a
set $I \subset \{1,\ldots,LN\}$ with $|I| = k$, we define $U_I$ to be the $N
\times k$ matrix obtained from $U$ by selecting the columns of $U$ corresponding
to the set $I$.

\begin{lemma}\label{lemma:norms_multiple}
In the setting of Theorem \ref{prop:uncertainty}, with probability tending to 1 as $N \to \infty$ (uniformly in $L \ge 2$), for
all $k \le LN - 1$,
\begin{equation}
\sqrt{\mathcal{S}_k} = \max_{|I| = k+1} \|U_I\|
\le 1 + \sqrt{\constantCg\frac{k+1}{N}\ln\Big(\frac{eNL}{k+1}\Big)},
\end{equation}
where $\constantCg$ is a universal constant.
\end{lemma}

\begin{proof}[Proof of Theorem \ref{prop:uncertainty}]
Let $\constantCg$ be the constant from Lemma \ref{lemma:norms_multiple}. Define $M_1 =
\Big(1 + \sqrt{\constantCg\frac{2}{N}\ln\Big(\frac{eLN}{2}\Big)}\Big)^2 - 1$
and for $2
\le i \le NL-1$,
\begin{equation}
M_i = \Big(1 +
\sqrt{\constantCg\frac{i+1}{N}\ln\Big(\frac{eLN}{i+1}\Big)}\Big)^2 - \Big(1 +
\sqrt{\constantCg\frac{i}{N}\ln\Big(\frac{eLN}{i}\Big)}\Big)^2 .
\end{equation}
We have for $2 \le i \le NL-1$,
\begin{equation}
\begin{split}
M_i &= 2 \left(\sqrt{\constantCg\frac{i+1}{N}\ln\Big(\frac{eLN}{i+1}\Big)} -
\sqrt{\constantCg\frac{i}{N}\ln\Big(\frac{eLN}{i}\Big)}\right)
+
\Big(\constantCg\frac{i+1}{N}\ln\Big(\frac{eLN}{i+1}\Big)-\constantCg\frac{i}{N}\ln\Big(\frac{eLN}{i}\Big)\Big)\\
& = 2\sqrt{\constantCg L} \left(f(\frac{i+1}{LN}) - f(\frac{i}{LN}) \right) +
\constantCg L \left(g(\frac{i+1}{LN}) - g(\frac{i}{LN}) \right),
\end{split}
\end{equation}
where
$f,g\colon [0,e] \to \R$ are given by $f(x) = \sqrt{x\ln(e/x)}$ and $g(x)= x\ln(e/x)$.

Both $f$ and $g$ are concave and thus
\begin{equation}
\begin{split}
M_i &\le 2\sqrt{\constantCg L}\frac{1}{LN}\frac{d}{dx} f\left(\frac{i}{LN}
\right) + \constantCg L\frac{1}{LN}\frac{d}{dx}g \left(\frac{i}{LN} \right)\\
&= 2\sqrt{\constantCg
L}\frac{1}{LN}\frac{\ln(\frac{LN}{i})}{2\sqrt{\frac{i}{LN}\ln(\frac{eLN}{i})}}
+ \frac{\constantCg L}{LN}\ln\left(\frac{LN}{i} \right)\\
& =
\frac{1}{N}\left(\frac{\sqrt{\constantCg}\ln(\frac{LN}{i})}{\sqrt{L}\sqrt{\frac{i}{LN}
 \ln\Big(\frac{eLN}{i}\Big)}} + \constantCg\ln\Big(\frac{LN}{i}\Big)\right) =:
\tilde{M}_i.
\end{split}
\end{equation}
Since
\begin{equation}
M_1+ \sum_{i=2}^{LN-1}\tilde{M}_i \ge \sum_{i=1}^{LN-1} M_i = \Big(1 +
\sqrt{\constantCg\frac{LN}{N}\ln\Big(\frac{eLN}{LN}\Big)}\Big)^2 - 1 > L-1,
\end{equation}
there exists maximum $N_0 < LN- 1$ such that
\begin{equation}
M_1 + \sum_{i=2}^{N_0} \tilde{M}_i \le L-1,
\end{equation}
(note that for $N$ sufficiently large, independent of $L$, $M_1 \le L-1$).

Define a vector $W \in \R^{N_0+1}$ by $W_1 = M_1$,
\begin{equation}
W_i = \tilde{M}_i,
\end{equation}
for $i = 2,\ldots,N_0$ and $W_{N_0+1} = L-1 -\sum_{i=1}^{N_0} W_i$.

Let $z_i$, $i =1,\ldots,NL$ be the non-increasing  rearrangement of
$p_1\oplus\cdots\oplus p_L$. Using Theorem \ref{th:strong-many_measurements} and
Lemma \ref{lemma:norms_multiple} we get that with probability tending to one as
$N \to \infty$,  for $k \le N_0+1$,
\begin{equation}
\begin{split}
z_1 + \ldots + z_k & \le \mathcal{S}_{k-1} 
\le \Big(1 + \sqrt{\constantCg\frac{k}{N}\ln\Big(\frac{eNL}{k}\Big)}\Big)^2\\
& = 1 + M_1 + \ldots+ M_{k-1} 
\le 1 + W_1 + \ldots + W_{k-1}.
\end{split}
\end{equation}
Also for $k > N_0+1$ we have $z_1 + \ldots+z_k \le L = 1 + W_1+\ldots+W_{N_0+1}$, so $z \prec 1\oplus R$ and as a consequence
\begin{equation}\label{eq:majorization-bound-mm}
\sum_{i=1}^L H(p_i) \ge -\sum_{i=1}^{N_0+1} W_i\ln W_i.
\end{equation}

Now, using the definition of $M_1$ and $\tilde{M}_i$, it is easy to see that
$M_1+ W_{N_0+1} \le  \constantCk\left(\sqrt{\frac{\ln(LN)}{N}} +
\frac{\ln(LN)}{N}\right)$. In particular for large $N$ (uniformly in $L \ge 2$)
we have $ M_1\ln M_1 + W_{N_0+1} \ln W_{N_0+1} \le \constantCl\ln^2 L$ and so

\begin{align}\label{eq:bound-on-HW}
-\sum_{i=1}^{N_0+1} W_i\ln W_i \ge &\sum_{i=1}^{N_0} - W_i \ln W_i\nonumber \\
 \ge & - \constantCl \ln^2 L + \left(\sum_{i=2}^{N_0} W_i \right)\ln N
 \nonumber\\
 & - \sum_{i=2}^{N_0}
 \frac{1}{N}\left(\frac{\sqrt{\constantCg}\ln(\frac{LN}{i})}{\sqrt{L}\sqrt{\frac{i}{LN}
  \ln(\frac{eLN}{i})}} + \constantCg\ln \left(\frac{LN}{i} \right)\right)
 \ln\left(\frac{\sqrt{\constantCg}\ln(\frac{LN}{i})}{\sqrt{L}\sqrt{\frac{i}{LN}
 \ln(\frac{eLN}{i})}} + \constantCg\ln\left(\frac{LN}{i} \right)\right)
 \nonumber\\
= & (L-1) \ln N - (M_1 + W_{N_0+1})\ln N -  \constantCl \ln^2 L - B_N
\nonumber\\
\ge & (L-1)\ln N - \constantCl\ln^2 L -
\constantCk\left(\sqrt{\frac{\ln(LN)}{N}} +
\frac{\ln(LN)}{N}\right)\ln N - B_N,
\end{align}
where
\begin{equation}
B_N = \sum_{i=2}^{N_0} \frac{1}{N}
\left(\frac{\sqrt{\constantCg}\ln(\frac{LN}{i})}{\sqrt{L}\sqrt{\frac{i}{LN}
\ln(\frac{eLN}{i})}} + \constantCg\ln\left(\frac{LN}{i} \right) \right)
\ln\left(\frac{\sqrt{\constantCg}\ln(\frac{LN}{i})}{\sqrt{L}\sqrt{\frac{i}{LN}
\ln(\frac{eLN}{i})}} + \constantCg\ln \left(\frac{LN}{i} \right)\right).
\end{equation}

Now, the following holds for a sufficiently large absolute constant
$\constantCm$. If
$i < LN/\constantCm$, then $\ln(\frac{eLN}{i}) \le (\frac{LN}{i})^{1/4}$ and
using
the inequality $L \ge 2$, we get
\begin{equation}
\left(\frac{\sqrt{\constantCg}\ln(\frac{LN}{i})}{\sqrt{L}\sqrt{\frac{i}{LN}
\ln(\frac{eLN}{i})}} + \constantCg\ln \left(\frac{LN}{i} \right)\right)
\ln\left(\frac{\sqrt{\constantCg}\ln(\frac{LN}{i})}{\sqrt{L}\sqrt{\frac{i}{LN}
\ln(\frac{eLN}{i})}} + C\ln \left(\frac{LN}{i} \right)\right)
\le \constantCn\Big(\frac{LN}{i}\Big)^{7/8},
\end{equation}
for some absolute constant $\constantCn$.

On the other hand if $i > LN/\constantCm$, then
\begin{equation}
\left(\frac{\sqrt{\constantCg}\ln(\frac{LN}{i})}{\sqrt{L}\sqrt{\frac{i}{LN}
\ln(\frac{eLN}{i})}} + \constantCg\ln \left(\frac{LN}{i} \right) \right)
\ln\left(\frac{\sqrt{\constantCg}\ln(\frac{LN}{i})}{\sqrt{L}\sqrt{\frac{i}{LN}
\ln(\frac{eLN}{i})}} + \constantCg\ln \left(\frac{LN}{i} \right)\right) \le
\constantCo,
\end{equation}
where $\constantCo$ is another absolute constant.

Thus, using $N_0 \le LN$, we get
\begin{equation}\label{eq:B_N-bounded}
\begin{split}
B_N &\le \frac{\constantCo}{N}  (N_0- LN/\constantCm)_+ +
\frac{\constantCn(LN)^{7/8}
}{N}\sum_{i=2}^{\lfloor LN/\constantCm\rfloor } \frac{1}{i^{7/8}}\\
&\le \constantCo L + \frac{\constantCn(LN)^{7/8}
}{N}\constantCp\Big(\frac{LN}{\constantCm}\Big)^{1/8} \le \constantCq L.
\end{split}
\end{equation}
It remains to bound from above the term $\constantCl\ln^2 L +
\constantCk\left(\sqrt{\frac{\ln(LN)}{N}} + \frac{\ln(LN)}{N}\right)\ln N$
appearing
in \eqref{eq:bound-on-HW}. It is easy to see that for sufficiently large $N$
(uniformly in $L\ge 2$) it is bounded by $\constantCr L$.

Combining this estimate with \eqref{eq:bound-on-HW} and \eqref{eq:B_N-bounded} gives
\begin{equation}
-\sum_{i=1}^{N_0+1} W_i\ln W_i \ge (L-1)\ln N - \constantCc L,
\end{equation}
with $\constantCc = \constantCq + \constantCr$.
By \eqref{eq:majorization-bound-mm} this ends the proof of the theorem.
\end{proof}


To relate the above result with earlier literature recall that
Wehner and Winter \cite{WW10} defined a function
\begin{equation}\label{eq:WW_conjecture1}
h(L) = \lim_{N \to \infty} \max_{U_1,\ldots,U_L \in U(N)} \frac{1}{\ln N} \min_{\psi} \frac{1}{L}\sum_{i=1}^L H(p^{(\psi,i)}).
\end{equation}
They ask whether $h(L) = 1 - \frac{1}{L}$ for all $L \ge 2$.
A related weaker question is whether there exists an increasing function
$f\colon \mathbb{N} \to
[0,\infty)$, such that $\lim_{L\to \infty} f(L) = \infty$ and
$h(L) \ge 1  - \frac{1}{f(L)}$.
Clearly any such function must be bounded from above by $L$.
Theorem~\ref{prop:uncertainty} immediately yields the following corollary.
\begin{corollary}\label{cor:WW-conjecture}
The conjecture by Wehner and Winter holds true, i.e. for every $L$ the limit
$h(L)$ in \eqref{eq:WW_conjecture1} exists and
\begin{equation}
h(L) = 1 - \frac{1}{L}.
\end{equation}
\end{corollary}

In \cite{Fawzi2011} Fawzi et al. showed  by probabilistic methods that for any
$L \ge 2$ and $N > 2$ there exist $L$  unitary matrices $U_1, \ldots,U_L$ such
that
\begin{equation}
\min_{\psi}\frac{1}{L}\sum_{i=1}^L H(p^{(\psi,i)}) \ge \Big(1 -
\sqrt{\frac{c'\ln L}{L}}\Big)\ln N - \ln\Big(\frac{18L}{c'\ln L}\Big) -
H\Big(\sqrt{\frac{c'\ln L}{L}},1-\sqrt{\frac{c'\ln L}{L}}\Big)
\end{equation}
for some universal constant $c'$. In particular this proves the weak form of the
Wehner and Winter conjecture with $f(L) = \sqrt{\frac{L}{c'\ln N}}$. We remark
that Fawzi et al. \cite{Fawzi2011} obtained also more explicit constructions of
matrices satisfying entropic uncertainty principles for $N$ being a power of
$2$, as their constructions can be efficiently performed by quantum circuits.
However the number of measurements $L$ in their scheme is  bounded from above by
a polynomial in $\ln N$.

Following the strategy of \cite{Fawzi2011}, our Theorem~\ref{prop:uncertainty} can
be directly applied to protocols of  locking of the classical information in
quantum states \cite{DHLST04,DFHL11}. Our bounds for the average entropy,  valid
for large dimension $N$ and an arbitrary number of measurements $L$, provide
more precise estimations concerning  the information leaked by a measurement
from a quantum system used in an information locking scheme.


\section{Concluding remarks}
\label{sec:concluding}

In this work we analyzed truncations of $N\times N$ random unitary matrices and
obtained estimations \eqref{trun1},  \eqref{eq:harmonic} and
\eqref{eq:uniform-one-times-n} for their norms. These results allowed us to
study various entropic uncertainty relations providing the bounds for the sum of
entropies describing information gained in two orthogonal measurements of any
$N$-dimensional pure quantum state.

Our analysis reveals in particular that classical relations, known to be optimal
in extremal settings, do not perform well in generic situations. For instance,
the Maassen--Uffink bound (\ref{MU}) averaged with the Haar measure over the
unitary group behaves asymptotically as $\ln N - \ln \ln N - \ln 2$. As the
largest element of a random orthogonal matrix is typically larger by a factor of
$\sqrt{2}$ \cite{Jiang}, the same bound averaged over the orthogonal group gives
$\ln N - \ln \ln N - 2\ln 2$. These results can be compared with implications of
the strong entropic uncertainty relation which, averaged over the unitary group
gives a lower bound $\ln N - C$, which is close to the best possible one.
Although the exact value of the optimal constant $C$ is still unknown, the
sandwich form (\ref{sandw}) implies that $C\in (0.42, 3.49)$.

It is natural to conjecture that if $U$ is drawn from the Haar measure on the
unitary group $U(N)$ and $D_N = \min_{\psi} (H(p^\psi) + H(q^\psi))$, then
there exists a limit
\begin{equation}
\lim_{N\to \infty} (\ln N - \E D_N).
\end{equation}

Strong majorization entropic uncertainty relations can be also formulated for
$L$ orthogonal measurements, determined by a collection of $L$ unitary matrices
of order $N$. Making use of bounds for the norms of their submatrices we
established an estimate (\ref{multi}), which implies that the sum of $L$
entropies behaves asymptotically as $(L-1)\ln N  -  {\rm const}$. This result,
holding for an arbitrary number $L$ of measurements, is up to an additive
constant compatible with the estimate (\ref{sandw}) valid for $L=2$. In
particular it allows us to answer completely an open question by Wehner and
Winter on asymptotic behavior of constants in entropic uncertainty relations for
many measurements as the dimension of the underlying Hilbert space tends to
infinity. Furthermore, these bounds can be used to quantify the information
leaked due to measurements from a quantum system, in which information locking
protocol is applied \cite{Fawzi2011}.

A natural open question is to find more precise estimations for these additive
constants determining the typical behavior of entropic uncertainty relations. To
get tighter bounds for the averaged relation (\ref{eqn:strong-relation}) one
would need to improve  the bounds for the average norms of the leading
truncations of random unitaries. Note that the bounds (\ref{eq:harmonic}) and
(\ref{trun1}) derived in this work can be considered as complementary: The
former one holds for $m=1$ and an arbitrary $n\in [1,N]$, while the latter one
works for any sizes $n$ and $m$ of the submatrix, but provides non-trivial
estimates if $n$ is small with respect to the matrix size $N$. Therefore, it is
tempting to believe that establishing a new family of bounds for the norms
$||\widehat{U}^{(n,m)}||$, which share advantages of both known results, would
allow one to improve the quality of the asymptotic entropic uncertainty
relations. We also mention that obtaining optimal bounds on norms of submatrices
of a random unitary matrix seems to be an interesting problem in its own rights,
with potential applications in Random Matrix Theory and Asymptotic Geometric
Analysis.

\medskip

It is a pleasure to thank Patrick Coles and {\L}ukasz Rudnicki for fruitful
discussions and helpful remarks. We appreciate numerous constructive suggestions
of the referee which allowed us to improve the work. This work was supported  by
the Grants number
 DEC-2012/05/B/ST1/00412 (RA and RL),
 DEC-2012/04/S/ST6/00400 (ZP) and DEC-2011/02/A/ST1/00119 (K{\.Z})
 of the Polish National Science Centre NCN
 and in part by the Transregio-12 project C4 of the Deutsche Forschungsgemeinschaft.

\appendix

\section{Proofs of estimates for norms of submatrices} \label{submatr2}
\subsection{Notation}

Before we proceed with the proofs let us gather here some (rather standard)
notation we are going to use.

\medskip

For $|x\rangle \in \C^n$, by $\|x\|$, we will denote its standard Euclidean
norm, i.e. $\|x\|=\sqrt{\<x|x\>}$. In the course of the proof we will often
encounter the Euclidean norm of $A|x\>$, where $A$ is a matrix and $|x\> \in
\C^n$. To shorten the notation we will denote it by $\|Ax\|$, i.e.
\begin{equation}\label{eq:Euclidean_norm}
\|Ax\| = \sqrt{\<x|A^\dag A|x\>}.
\end{equation}
Recall also that if $A$ is a matrix, by $\|A\|$ we denote the operator norm of
$A$. We will also use the Hilbert-Schmidt norm of $A$, defined as
$\|A\|_{HS} = \sqrt{\tr AA^\dag}.$
By $|I|$ we will denote the cardinality of a finite set $I$.
For a positive integer $n$ by $S^{n-1}$ we will denote the unit sphere in $\R^n$
equipped with the standard Euclidean norm, while $S_\C^{n-1}$ will denote the
unit sphere in $\C^n$. Clearly $S_\C^{n-1}$ is isometric to $S^{2n-1}$.
By $\Delta_{n-1}$ we will denote the standard $(n-1)$-dimensional simplex in
$\R^n$, i.e.
\begin{equation}
\Delta_{n-1} = \Big\{ x = (x_1,\ldots,x_n) \in \R^n \colon \forall_{1\le i\le n}
\; x_i \ge 0\; \textrm{and}\; \sum_{i=1}^n x_i = 1 \Big\}.
\end{equation}

We will sometimes use the $O$ notation. For two sequences $(a_N)_{N\ge 1}$ and
$(b_N)_{N \ge 1}$ we will write $a_N = O(b_N)$ if there exists a constant $K$
such that for all $N \ge 1$, $a_N \le Kb_N$.
We recall that by $\prec$ we denote the majorization relation defined in Section
\ref{sec:II} after formula \eqref{eq:definition_of_R}.
\subsection{Proof of Theorem \ref{th:upper}}
Recall that a probability measure $\mu$ on a metric space $(X,d)$ satisfies a
log-Sobolev inequality with constant $C$ if for any locally Lipschitz
function $f$
\begin{equation}
\int f^2\ln f^2 d\mu-\int f^2 d\mu \ln \int f^2d\mu\leq 2C\int |\nabla f|^2d\mu,
\end{equation}
where $|\nabla f|$ is the length of gradient with respect to the metric $d$, i.e.
\begin{equation}
|\nabla f|(x) = \limsup_{y \to x} \frac{|f(y) - f(x)|}{d(x,y)}
\end{equation}
(see e.g. Chapter 3.1. of \cite{Le} or the Appendix of \cite{MM}). For any such
measure and any $L$-Lipschitz function $F$ we then have (cf. \cite[Section 5.1]{Le})
\begin{equation}
\int \exp\Big(\lambda \Big(F-\int Fd\mu\Big)\Big)d\mu \leq
\exp\Big(\frac{CL^2}{2}\lambda^2\Big)\quad \mbox{ for all }\lambda \in \er
\end{equation}
and
\begin{equation}
\label{conc}
\mu\Big(F\geq \int Fd\mu+t\Big)\leq \exp\Big(-\frac{t^2}{2CL^2}\Big)\quad \mbox{ for }t\geq 0.
\end{equation}

We will use the following estimate of the log-Sobolev constant for the unitary
group (cf. \cite[Theorem 15]{MM}).
\begin{theorem}\label{thm:MeckesMeckes}
The Haar measure on the unitary group $U(N)$ satisfies a log-Sobolev inequality
with constant $6/N$ with respect to the Hilbert-Schmidt distance.
\end{theorem}

We recall, that for a $N\times N$ matrix $U=(U_{i,j})_{i,j=1}^N$ by
$\|\widehat{U}^{(n,m)}\| $ we denote the maximal norm of its $n\times m$
submatrices, i.e
\begin{equation}
\|\widehat{U}^{(n,m)}\| :=\max_{|I|=n,|J|=m}\|U(I,J)\|,
\end{equation}
where $U(I,J):=(U_{i,j})_{i\in I,j\in J}$.

\begin{proof}[Proof of Theorem~\ref{th:upper}]
The function $U\mapsto \|\widehat{U}^{(m,n)}\|$ is $1$-Lipschitz with respect to
the Hilbert-Schmidt norm. Therefore the estimate \eqref{conc_Ank} immediately
follows by \eqref{conc} and Theorem \ref{thm:MeckesMeckes}.

Observe that for any $|x\>\in \ce^N$ with $\|x\|=1$, the random variable $U|x\>$
is uniformly distributed on $S_\ce^{N-1}\simeq S^{2N-1}$. It is well known that
for any $l$, the uniform distribution on $S^{l}$ satisfies log-Sobolev
inequality with constant $1/l$, (cf. formula (5.7) in \cite{Le}). For any
$|y\>\in \ce^N$ the function $z\mapsto \Re \langle y|z\rangle$ is
$\|y\|$-Lipschitz on $S^{N-1}$. Therefore, using the fact that $\E \langle y|U|
x\rangle = 0$, we get
\begin{equation}
\Ex e^{\lambda \Re \langle y|U|x\rangle} \leq
\exp\Big(\frac{1}{2(2N-1)}\lambda^2\Big)\quad \mbox{ for all }\lambda \in \er,\ |x\>,|y\>\in S_\ce^{N-1}.
\end{equation}

Now suppose that we have a finite set $E\subset S_\ce^{N-1}\times S_\ce^{N-1}$. Then
\begin{equation}
\label{maxE}
\Ex \max_{(|x\>,|y\>)\in E}\Re  \langle y|U|x\rangle \leq \sqrt{\frac{2\ln |E|}{2N-1}}.
\end{equation}
Indeed we have for $\lambda>0$,
\begin{equation}
\Ex\exp\Big(\lambda \max_{(|x\>,|y\>)\in E}\Re \langle y|U|x\rangle\Big)
\leq \Ex \sum_{(|x\>,|y\>)\in E} e^{\lambda \Re \langle y|U|x\rangle }
\leq |E|\exp\Big(\frac{1}{2(2N-1)}\lambda^2\Big).
\end{equation}
Jensen's inequality gives
\begin{equation}
\Ex \exp\Big(\lambda \max_{(|x\>,|y\>)\in E}\Re \langle y|U|x\rangle\Big)\geq
\exp\Big(\lambda \Ex \max_{(|x\>,|y\>)\in E}\Re \langle y|U|x\rangle\Big),
\end{equation}
hence
\begin{equation}
\Ex \max_{(|x\>,|y\>)\in E}\Re \langle y|U|x\rangle \leq\inf_{\lambda>0}\frac{1}{\lambda}\Big(\ln |E|+\frac{1}{2(2N-1)}\lambda^2\Big)
=\sqrt{\frac{2\ln |E|}{2N-1}}.
\end{equation}

Let us now estimate $\Ex \|\widehat{U}^{(m,n)}\|$.  For any $\emptyset\neq I\subset \{1,\ldots,N\}$  consider the ($|I|$-1)-dimensional unit sphere
\begin{equation}
S_I:=\{|x\> = (x_1,\ldots,x_N) \in S_\ce^{N-1}\colon  x_i = 0\;\textrm{for}\; i\notin I\}
\end{equation}
and choose an $\ve$-net $E_I$ in $S_I$ of cardinality at most $(1+2/\ve)^{2|I|}$
(such a net exists by standard volumetric estimates, see e.g. \cite{Pisier}).
Let $E_l:=\bigcup_{|I|=l}E_I$, then for any $1\leq l\leq N$,
\begin{equation}
|E_l|\leq \binom{N}{l}\Big(1+\frac{2}{\ve}\Big)^{2 l}\leq \Big(\frac{eN}{l}\Big)^l\Big(1+\frac{2}{\ve}\Big)^{2l}.
\end{equation}
Estimate \eqref{maxE} gives
\begin{equation}
\Ex\max_{|x\>\in E_n,|y\>\in E_m}\Re \langle y|U| x\rangle \leq
\sqrt{\frac{2}{2N-1}(\ln |E_n|+\ln|E_m|)}.
\end{equation}
Finally it is not hard to see that
\begin{equation}
\|\widehat{U}^{(m,n)}\| \leq \frac{1}{1-2\ve-\ve^2}\max_{|x\>\in E_n,|y\>\in E_m}\Re \langle y|U|x\rangle.
\end{equation}
Inequality \eqref{eq:norm_expectation_estimate} follows now easily by the three
last estimates. The bound \eqref{trun1} follows from
\eqref{eq:norm_expectation_estimate} by elementary calculations.
\end{proof}

\subsection{Proof of Theorem \ref{th:submatrices-fixed-size}}\label{sec:derivation1}

Note that the upper bound on $\|\widehat{U}^{(n,m)}\|$ follows from the already
proven Theorem \ref{th:upper}. To complete the proof it is thus enough to show
that for all fixed positive integers $n,m$ and $\varepsilon > 0$,
\begin{align}\label{eq:goal}
\p\Big(\|\widehat{U}^{(n,m)}\| \le (1-\varepsilon)\sqrt{(n+m)\frac{\ln N}{N}}\Big) \to 0 ,
\end{align}
as $N\to \infty$.

Let $\Gamma = (\Gamma_{ij})_{i,j=1}^N$ be a $N\times N$ matrix whose entries
are i.i.d. standard complex Gaussian variables (i.e. their real and imaginary
parts are independent, with Gaussian distribution of mean zero and variance
$1/2$, or equivalently with the density $g(x) = \frac{1}{\sqrt{\pi}}e^{-x^2}$
with respect to the Lebesgue measure on $\R$).

Set $M = M_N = N/\ln^2 N$. By Theorem 6 in \cite{Jiang} (applied with $m = M_N$,
$r = 1/\ln N$, $s = \ln N/(\ln\ln N)^{1/2}$, $t = \sqrt{\ln N/\ln\ln N}$, cf.
formula (2.10) in \cite{Jiang}) we can assume that
\begin{equation}\label{eq:A19}
\max_{i\le N,j\le M_N}|\sqrt{N}U_{ij} - \Gamma_{ij}|\le
\constantCd
\sqrt{\ln N/\ln\ln  N},
\end{equation}
with probability at least
$1 - \constantCe \exp(-\ln^{3/2} N)$.


By 
\eqref{eq:A19}
it is enough to show that with probability
tending to 1,
\begin{equation}
\|\widehat{\Gamma'}^{(n,m)}\| \ge (1-\varepsilon)\sqrt{(n+m)\ln N},
\end{equation}
where $\Gamma'$ is the $M_N\times M_N$ principal submatrix of $\Gamma$.

Since $\frac{\ln M_N}{\ln N} \to 1$ as $N \to \infty$, \eqref{eq:goal} will
follow if we prove
\begin{proposition} For any positive integers $n,m$ and any $\varepsilon > 0$,
\begin{align}\label{eq:g_goal}
\p\Big(\|\widehat{\Gamma}^{(n,m)}\| \le (1-\varepsilon)\sqrt{(n+m)\ln N}\Big) \to 0.
\end{align}
\end{proposition}

\begin{proof}
First note that by the concentration property of Gaussian measures (see e.g.
\cite{Le}) and the fact that $\|\widehat{\Gamma}^{(n,m)}\|$ is 1-Lipschitz with
respect to the Hilbert-Schmidt norm, we have
\begin{equation}
\p\Big(\Big|\|\widehat{\Gamma}^{(n,m)}\| - \Ex \|\widehat{\Gamma}^{(n,m)}\| \Big| \ge t\Big) \le 2\exp(-t^2).
\end{equation}
Thus to prove the proposition it is enough to show that for every $\varepsilon >
0$, and $N$ large enough $\Ex \|\widehat{\Gamma}^{(n,m)}\| \ge
(1-\varepsilon)\sqrt{(n+m)\ln N}$. Assume that $\Ex \|\widehat{\Gamma}^{(n,m)}\|
< (1-\varepsilon)\sqrt{(n+m)\ln N}$. Then, again by concentration
$\p(\|\widehat{\Gamma}^{(n,m)}\| \ge (1-\varepsilon/2)\sqrt{(n+m)\ln N}) \le
1/N^{(n+m)\varepsilon^2/4} \to 0$ as $N \to \infty$. Therefore, to prove
\eqref{eq:g_goal} it is enough to show that for every $\varepsilon > 0$, there
exists $d > 0$ such that for $N$ large enough, we have
\begin{align}\label{eq:reduction}
\p\Big(\|\widehat{\Gamma}^{(n,m)}\| \ge (1-\varepsilon)\sqrt{(n+m)\ln N}\Big) > d.
\end{align}
It is well known that $|\Gamma_{ij}|^2$ are standard exponential variables (i.e.
they have a density $g(x) = e^{-x}\mathbf{1}_{[0,\infty)}(x)$), therefore
\begin{align}
\p\Big(|\Gamma_{ij}|^2 \ge \frac{n+m}{nm}\ln N\Big) = \frac{1}{N^{\frac{n+m}{nm}}}.
\end{align}
Moreover $\Gamma_{ij}$ are rotationally invariant, so for any $\delta \in (0,1)$,
\begin{align}\label{eq:single_entry}
\p\Big(|\Gamma_{ij}|^2\ge \frac{n+m}{nm}\ln N, \Arg \Gamma_{ij} \in [0,2\pi \delta)\Big) = \frac{\delta}{N^{\frac{n+m}{nm}}}.
\end{align}
Consider any $I,J \subset \{1,\ldots,N\}$ with $|I| = n$, $|J| = m$ and define
the event
\begin{equation}
\calE(I,J) = \Big\{ \forall_{i\in I, j \in J} \; |\Gamma_{ij}|^2 \ge \frac{n+m}{nm}\ln N, \Arg \Gamma_{ij} \in [0,2\pi \delta)\Big\}.
\end{equation}
Note that for $\delta$ small enough, depending on $\varepsilon$, on the event
$\calE(I,J)$ we have
\begin{align}\label{eq:implication}
\|\Gamma(I,J)\|\ge (1-\varepsilon)\sqrt{(n+m)\ln N},
\end{align}
where $\Gamma(I,J) = (\Gamma_{ij})_{i\in I,j\in J}$.

Indeed for the unit vector $|z\> = m^{-1/2}(1,\ldots,1) \in \C^m$ we have
(recall our notation introduced in \eqref{eq:Euclidean_norm})
\begin{equation}
\|\Gamma(I,J)z\|^2 = \<z|\Gamma(I,J)^\dag\Gamma(I,J)|z\> = \sum_{i\in I} \frac{1}{m}\Big|\sum_{j\in J} \Gamma_{ij}\Big|^2.
\end{equation}
Now for $\delta$ small enough,
\begin{align}
\Big|\sum_{j\in J} \Gamma_{ij}\Big|^2
&= \sum_{j,j' \in J} |\Gamma_{ij}||\Gamma_{ij'}|
   \cos \left(\Arg \Gamma_{ij}\overline{\Gamma_{i,j'}}\right) \nonumber\\
&\ge \sum_{j,j' \in J} |\Gamma_{ij}||\Gamma_{ij'}|\cos \left(2\pi \delta\right)
 \ge (1-\varepsilon)^2m^2\frac{n+m}{nm}\ln N
\end{align}
and thus
\begin{equation}
\|\Gamma(I,J)\|^2 \ge \|\Gamma(I,J)z\|^2 \ge (1-\varepsilon)^2(n+m)\ln N,
\end{equation}
which proves \eqref{eq:implication}.

Thus to prove the proposition it is enough to show that for $N$ large enough,
\begin{align}\label{eq:to_prove}
\p\Big(\bigcup_{|I| = n, |J| = m} \calE(I,J)\Big) \ge d.
\end{align}
By the Bonferroni inequality we have
\begin{align}
\p\Big(\bigcup_{|I| = n, |J| = m} \calE(I,J)\Big) &\ge \sum_{|I| = n, |J| = m} \p(\calE(I,J))
- \sum_{{|I| = |I'|= n, |J| = |J'| =  m}\atop {(I,J) \neq (I',J')}}\p(\calE(I,J)\cap\calE(I',J'))\\
&=:A-B.\nonumber
\end{align}
By \eqref{eq:single_entry} and independence of the entries of $\Gamma$,
\begin{equation}
A=\binom{N}{m}\binom{N}{n}  \frac{\delta^{mn}}{N^{n+m}} \to \frac{\delta^{nm}}{n!m!} ,
\end{equation}
as $N \to \infty$.

Now we group the summands in $B$, depending on the cardinality of $I\cup I'$ and
$J\cup J'$ and obtain
\begin{align}\label{eq:grouping}
B=\sum_{{n\le r \le 2n, m\le s\le 2m} \atop {r + s > n+m}} \sum_{{|I| = |I'|= n, |J| = |J'| =  m}\atop {|I\cup I'| = r, |J\cup J'| = s}}\p(\calE(I,J)\cap\calE(I',J')).
\end{align}
For fixed $r,s$ there are at most $C_{rs}N^{r+s}$ pairs $(I,J),(I',J')$ such
that $|I| = |I'|= n$, $|J| = |J'| =  m$,  $|I\cup I'| = r, |J\cup J'| = s$ where
$C_{rs}$ is a constant depending only on $r$ and $s$. For each such pair the
event $\calE(I,J)\cap\calE(I',J')$ is the intersection of $rs - 2(r-n)(s-m)$
independent events of the form \eqref{eq:single_entry}. Therefore,
\begin{equation}
\p(\calE(I,J)\cap\calE(I',J')) = \delta^{rs - 2(r-n)(s-m)}N^{-{(rs-2(r-n)(s-m))(n+m)} / {(nm)}}
\end{equation}
and as a consequence
\begin{align}
\sum_{{|I| = |I'|= n, |J| = |J'| =  m}\atop {|I\cup I'| = r, |J\cup J'| = s}}
\p(\calE(I,J)\cap\calE(I',J'))
&\le
C_{rs}\delta^{rs - 2(r-n)(s-m)}N^{r+s - {(rs-2(r-n)(s-m))(n+m)}/{(nm)}}\\
&= C_{rs}\delta^{rs - 2(r-n)(s-m)}N^{{(s-2m)(r-n)}/{n} + {(r-2n)(s-m)}/{m}}.\nonumber
\end{align}
One can see that if $r \neq 2n$ or $s \neq 2m$ then
$\frac{(s-2m)(r-n)}{n}+\frac{(r-2n)(s-m)}{m} < 0$ and so the contribution to
\eqref{eq:grouping} from such pairs converges to $0$ as $N \to \infty$.
Therefore, for $\delta$ small enough and large $N$,
\begin{equation}
\p(\bigcup_{|I| = n, |J| = k} \calE(I,J)) \ge \frac{\delta^{nm}}{2n!m!} -
C_{2n,2m}\delta^{2nm} > \frac{\delta^{nm}}{4n!m!}.
\end{equation}
Thus \eqref{eq:to_prove} holds with $d = \frac{\delta^{nm}}{4n!m!}$, which ends
the proof of the proposition.
\end{proof}

\subsection{Proof of Theorems \ref{th:single-column} and \ref{th:one-times-n}}

\begin{proof}[Proof of Theorem \ref{th:single-column}]
Note that the  first column $(U_{1i})_{i=1}^N$ (or any other column or row of
$U$) is uniformly distributed on $S_\ce^{N-1}\simeq S^{2N-1}$. Hence the squares
of the moduli of its entries, $q_i=|U_{1i}|^2$, $i=1,\dots N$, form a random
probability vector uniformly distributed on the simplex $\Delta_{N-1} \subset
{\mathbb{R}}^{N}$ (this observation seems to be a part of the folklore, it can
be easily obtained by 1) expressing the uniform measure on $S_\ce^{N-1}$ in
terms of normalized complex Gaussian vectors, 2) using the fact that the square
of the absolute value of a standard complex Gaussian variable has standard
exponential distribution, 3) invoking the well known fact that a self normalized
vector with i.i.d. standard exponential coordinates is distributed uniformly on
$\Delta_{N-1}$, see e.g. \cite{Kotz2000}).

To look for the largest component of the vector we order $q_1,\ldots,q_N$ in a
weakly decreasing order, $q_1^{\downarrow} \ge q_2^{\downarrow} \ge \dots \ge
q_N^{\downarrow}$. It is not hard to notice that the random vector
$q^{\downarrow}=(q_1^{\downarrow}, q_2^{\downarrow}, \dots, q_N^{\downarrow})$
is uniformly distributed on the simplex  ${\tilde \Delta_{N-1}}$ with vertices
$(1,0,\dots,0)$, $\frac{1}{2}(1,1,0,\dots,0)$,\dots,$\frac{1}{N}(1,1,\dots,1)$.

\begin{figure}[htbp]
\centering
\includegraphics[width=0.75\linewidth]{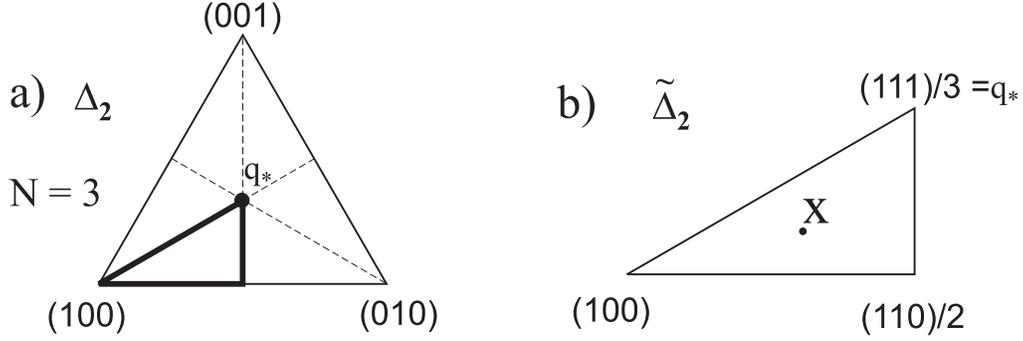}
\caption{
a) Simplex $\Delta_2$  for $N=3$ centered at $q^*=\left(\frac13,\frac13,\frac13\right)$ and b) its asymmetric part ${\tilde\Delta_2}$.
The barycenter $X$ of ${\tilde\Delta_2}$  
with components $\left(\frac{11}{18},\frac{5}{18},\frac{2}{18}\right)$ represents
the averaged ordered vector $\Ex q^{\downarrow}$.
}
\label{fig:simpl3}
\end{figure}

Thus the mean value of $q^{\downarrow}$ is the barycenter of ${\tilde
\Delta_{N-1}}$. Its coordinates can be expressed in terms of the {\sl harmonic}
numbers $H_m:=\sum_{j=1}^m 1/j$, which asymptotically behave as $\ln m +\gamma$,
where $\gamma\approx 0.5772$ denotes the Euler constant. Namely,
\begin{equation}
\Ex q_m^{\downarrow}  = \frac{1}{N} \sum_{j=m}^N \frac{1}{j} = \frac{1}{N}(H_N-H_{m-1}).
\label{xm}
\end{equation}

Denote by $X_{n,i}$ ($i = 1,\ldots,N$) the maximum norm of a subvector of
dimension $n\leq N$ of the $i$-th column of $U$. The average of $X_{n,i}^2$ is
equal to the sum of the first $n$ components of the ordered vector
$q^{\downarrow}$, averaged over the simplex ${\tilde \Delta_{N-1}}$, i.e.
\begin{equation}
  \Ex X_{n,i}^2 =
 \Ex \sum_{m=1}^n q^\downarrow_m  = \frac{1}{N} \sum_{m=1}^n  \sum_{j=m}^N \frac{1}{j} \ .
\label{maxA1}
\end{equation}
To evaluate this sum we divide the summation region in the $(m,j)$ plane into a
triangle and a rectangle and change the summation order,
 \begin{equation}
 \Ex X_{n,i}^2
 = \frac{1}{N}
 \Bigl[  \sum_{j=1}^n \frac{1}{j} \sum_{m=1}^j 1 +
       \sum_{j=n+1}^N \frac{1}{j} \sum_{m=1}^n 1 \Bigr] =
 \frac{n}{N}  \Bigl[ 1+ \sum_{j=n+1}^N \frac{1}{j} \Bigr]
 =\frac{n}{N} \Bigl[ 1+  H_N -H_n\Bigr].
\label{maxA2}
\end{equation}
We can now easily finish the proof, since we have
\[
||\widehat{U}^{(n,1)}||=\max_{i\leq N}X_{n,i},
\]
and \eqref{eq:harmonic} follows by the concentration of measure (recall that the
uniform distribution on $S_\ce^{N-1}$ satisfies the log-Sobolev inequality with
constant $1/(2N-1)$).
\end{proof}

\begin{proof}[Proof of Theorem \ref{th:one-times-n}]
Let us fix $\varepsilon > 0$ and let $n_0 = n_0(\varepsilon)$ be a sufficiently
large constant depending on $\varepsilon$, to be chosen later on. By Theorem
\ref{th:submatrices-fixed-size}, with probability tending to 1 as $N \to
\infty$, we have for all $n \le n_0$
\begin{equation}\label{eq:small-size}
(1-\varepsilon)\sqrt{\frac{n+1}{N}(1+  \ln \Big(\frac{N}{n}\Big)\Big) } \le \|\widehat{U}^{(n,1)}\| \le (1+\varepsilon)\sqrt{\frac{n+1}{N}(1+  \ln \Big(\frac{N}{n}\Big)\Big) }.
\end{equation}
(note that in this range of $n$, $(1+  \ln (N/n)) = (1+o(1))\ln N$ as $N \to
\infty$).

Consider any $n \ge n_0$. As in the proof of Theorem \ref{th:single-column},
denote by $X_{n,i}$ ($i = 1,\ldots,N$) the maximum norm of a $n\times 1$
submatrix of the $i$-th column of $U$.

Using \eqref{maxA2} we get $\E X_{n,i}^2 = \frac{n}{N}(1 + H_N - H_n)$ and so
\begin{equation}
(1-\varepsilon/8)^2 \frac{n+1}{N}(1 + \ln(N/n)) \le \E X_{n,i}^2
\le (1+\varepsilon/8)^2 \frac{n+1}{N}(1 + \ln(N/n)),
\end{equation}
where we used the fact that $n > n_0$.

Now, by integration by parts and \eqref{conc} it is easy to see that for large
$N$,
\begin{equation}
\E X_{n,i} \ge  \sqrt{\E X_{n,i}^2} - O(1/\sqrt{N}) \ge (1 - \varepsilon/8)\sqrt{\E X_{n,i}^2},
\end{equation}
where the second inequality holds for $n > n_0$ and $n_0$ large enough. Thus
\begin{equation}
(1-\varepsilon/2) \sqrt{\frac{n+1}{N}(1 + \ln(N/n))}
\le \E X_{n,i} \le (1+\varepsilon/2)\sqrt{ \frac{n+1}{N}(1 + \ln(N/n))}.
\end{equation}

Now, using again \eqref{conc} together with the union bound we get
\begin{equation}\label{eq:whatever}
(1-\varepsilon) \sqrt{\frac{n+1}{N}(1 + \ln(N/n))} \le X_{n,i}
\le (1+\varepsilon)\sqrt{ \frac{n+1}{N}(1 + \ln(N/n))},
\end{equation}
for all $n > n_0$ and $i \le N$, with probability at least
\begin{equation}
1 - N\sum_{n=n_0+1}^N\exp\Big(-\frac{\varepsilon^2}{24} (n+1)\ln(eN/n)\Big),
\end{equation}
which can be made arbitrarily close to one for $N \to \infty$ if one chooses
$n_0(\varepsilon)$ sufficiently large (as can be easily seen by splitting the
second term  into two separate sums over  $n_0< n < \sqrt{N}$ and $\sqrt{N}\le n
\le N$ respectively).
The proof is concluded by combining \eqref{eq:small-size} and
\eqref{eq:whatever}.
\end{proof}


\section{Proof of proposition~\ref{prop:uniform-bound-for-sk}} \label{sec:uniform-bound-for-sk}

\begin{proof}[Proof of Proposition \ref{prop:uniform-bound-for-sk}]
Note, that if for any $n,m$, such that $n+m = k+1$, we have
\begin{align}\label{eq:local_goal}
\E \|\widehat{U}^{(n,m)}\| \le \sqrt{ D \frac{k+1}{N} \left( 1 +
\ln\left(\frac{2 N}{k+1}\right) \right) },
\end{align}
then \eqref{eq:basic} holds with $\constantCf = D + \delta$, for any $\delta >
0$, since
by Theorem \ref{thm:MeckesMeckes} and \eqref{conc} we get
\begin{equation}
\begin{split}
\p &\left( \exists_{k \le N}\, s_k > \sqrt{\constantCf\frac{k+1}{N}\left(1 +
\ln\left(\frac{2 N}{k+1}\right)\right)} \right)
\le \sum_{k=1}^{N-1} \sum_{n=1}^k
\p \left( \|\widehat{U}^{(n,k+1-n)}\| > \sqrt{\constantCf\frac{k+1}{N}\left( 1
+ \ln\left(\frac{2N}{k+1}\right)\right)}\right) \\
& \ \ \le \sum_{k=1}^N k\exp\left(-c_{D,\delta}N \frac{k+1}{N} \left( 1 + \ln\left(\frac{2 N}{k+1}\right)\right) \right) \to 0 \text{ as } N \to \infty.
\end{split}
\end{equation}
By Theorem~\ref{th:upper} we get
\begin{equation}\label{eq:yet_another}
\E \|\widehat{U}^{(n,m)}\| \leq
\frac{1}{1-2\ve-\ve^2}\sqrt{\frac{2}{2N-1}}\left(m\ln \frac{eN}{m}+n\ln\frac{eN}{n}+2(n+m)\ln \left(1+\frac{2}{\ve} \right)\right)^{1/2},
\end{equation}
for $\varepsilon < 1/3$.

Note that when $k+1 > N/D$, the right hand side of \eqref{eq:local_goal} exceeds
1, so the inequality is satisfied trivially. We can therefore assume that $k+1
\le N/D$. We maximize the right hand side of  \eqref{eq:yet_another} under the
constraint $n + m = k+1$ and get
\begin{equation}
\begin{split}
\E \|\widehat{U}^{(n,m)}\| &\leq
\frac{1}{1-2\ve-\ve^2}\sqrt{\frac{2}{2N-1}}\Big((k+1) \ln \frac{2eN}{k+1}+ 2(k+1)\ln(1+\frac{2}{\ve})\Big)^{1/2}\\
&\le \sqrt{\frac{2}{2N-1}}\Big((k+1) \ln\frac{2eN}{k+1}\Big)^{1/2}
 \Big(\frac{1}{(1-2\ve-\ve^2)^2}\Big(1 + \frac{2\ln\Big(1+\frac{2}{\ve}\Big)}
 {\ln(2eD)}\Big)\Big)^{1/2},
\end{split}
\end{equation}
where we used the assumption $N/(k+1) \ge D$.

Now we set $D = 4.175$ and perform a minimization with respect to $\varepsilon
\in (0,1/3)$ of the expression
\begin{equation}
\frac{1}{(1-2\ve-\ve^2)^2} \left(1 + \frac{2\ln(1+\frac{2}{\ve})}{\ln (2eD)} \right).
\end{equation}
The numerical value of the minimum is approximately $4.172 \le 4.175$ (obtained
for $\varepsilon = 0.039$). This shows \eqref{eq:local_goal} with $D = 4.175$
and thus the proposition holds true with $\constantCf = 4.18$.
\end{proof}


\section{Proof of Lemma \ref{lemma:norms_multiple}}\label{sec:lemma:norms_multiple}

\begin{proof}[Proof of Lemma \ref{lemma:norms_multiple}]
Recall the notation introduced in equation \eqref{eq:Euclidean_norm}. Let us
note that by Theorem \ref{thm:MeckesMeckes} and the tensorization property of entropy, $U$ satisfies the
log-Sobolev inequality with parameter $6/N$ with respect to the
Hilbert-Schmidt metric. In particular, since for any unit vector $|x\> =
(x_1,\ldots,x_{NL}) \in \C^{NL}$, the map $U \mapsto \|Ux\|$ is 1-Lipschitz, we
get
\begin{equation}\label{eq:many-conc}
\p(\|Ux\| \ge \E \|Ux\| + t) \le 2e^{-Nt^2/12}.
\end{equation}

Denote the columns of $U$ by $|Y_i\>$, $i = 1,\ldots,NL$.
We also have
\begin{equation}
\E \|Ux\|^2 = \E \langle x|U^\dag U|x\rangle
= \sum_{i=1}^{NL} |x_i|^2 \E \|Y_i\|^2 +
\sum_{i\neq j} x_i\bar{x}_j \E \langle Y_i|Y_j\rangle = 1,
\end{equation}
where we used the fact that for each $i \neq j$, $|Y_i\>$ and $|Y_j\>$ are of mean
zero and either stochastically independent or orthogonal with probability one.
Thus $\E\|Ux\| \le 1$. Moreover, by \eqref{eq:many-conc} and integration by parts
\begin{equation}
1 = \sqrt{\E\|Ux\|^2} \le \E\|Ux\| + \sqrt{\frac{24}{N}}.
\end{equation}

Consider now a fixed set $I \subset \{1,\ldots,NL\}$ with $|I| = k+1$ and let
$\mathcal{N}_I$ be a $1/4$-net in the unit ball of $\C^I = \{|x\> =
(x_1,\ldots,x_{NL}) \in \C^{NL} \colon x_i = 0 \textrm{\; for $i \notin I$}\}$
of cardinality $10^{2(k+1)}$ (it exists by standard volumetric estimates, see
\cite{Pisier}). If $\constantCh$ is a sufficiently large absolute constant,
then by the
union bound, with probability at least
\begin{equation}
1 - 10^{2(k+1)}\binom{LN}{k+1}e^{-\constantCh(k+1)\ln(eNL/(k+1))/13} \ge 1 -
e^{-(k+1)\ln(eNL/(k+1))},
\end{equation}
we have
\begin{equation}\label{eq:net}
1 - \sqrt{\constantCh\frac{k+1}{N}\ln\Big(\frac{eNL}{k+1}\Big)} \le \|Ux\| \le
1 + \sqrt{\constantCh\frac{k+1}{N}\ln\Big(\frac{eNL}{k+1}\Big)} ,
\end{equation}
for all $I$ with $|I| = k+1$ and $|x\> \in \mathcal{N}_I$.

Let $\delta = \sqrt{\constantCh\frac{k+1}{N}\ln\Big(\frac{eNL}{k+1}\Big)}$. If
$\delta >
1$, then the second inequality in \eqref{eq:net} implies that
\begin{equation}
\|U_I\| \le \frac{4}{3}(1+\delta) \le 1 +
\sqrt{\constantCi\frac{k+1}{N}\ln\Big(\frac{eNL}{k+1}\Big)} ,
\end{equation}
for $\constantCi$ sufficiently large (depending only on $\constantCh$).

If $\delta < 1$, then on the event where \eqref{eq:net} holds, we have for $I$
with $|I| = k+1$ and $|x\> \in \mathcal{N}_I$,
\begin{equation}
1 - 2\delta \le \|U_Ix\|^2 = \langle x|U_I^\dag U_I |x\rangle \le 1 + 3\delta ,
\end{equation}
which implies that the operator $A = U_I^\dag U_I - \id$ on $\C^{I}$ (where $\id$ is the identity matrix), satisfies
\begin{equation}
|\langle x|A|x \rangle| \le 3\delta ,
\end{equation}
for $|x\> \in \mathcal{N}_I$.

Let now $|y\>$ be any unit vector in $\C^I$ and $|x\>$ a point in
$\mathcal{N}_I$ such that
$\big\| \ket{x} - \ket{y}\big\| < 1/4 $. We have
\begin{equation}
\big|\langle y|A|y\rangle \big|
\le
\big|\left(\langle y| - \langle x|\right) A \left(| y \rangle - | x \rangle \right) \big|
+
\big|\langle x| A \left(| y \rangle - | x \rangle \right) \big|
+
\big |\left(\langle y| - \langle x|\right)  A|x\rangle \big|
+
\big|\langle x|A|x\rangle \big|
\le
\frac{1}{16}\|A\|
+
\frac{1}{2}\|A\|
+
3\delta.
\end{equation}

Taking the supremum over $|y\>\in \mathbb{S}_\C^{N-1}$, using the fact that $A$
is Hermitian and performing easy calculations we get
\begin{equation}
\|A\| \le 7\delta,
\end{equation}
which implies that
\begin{equation}
\|U_I\|^2 \le 1 + 7\delta
\end{equation}
and as a consequence
\begin{equation}
\|U_I\| \le 1 + \sqrt{\constantCj\frac{k+1}{N}\ln\Big(\frac{eNL}{k+1}\Big)}.
\end{equation}
Now it remains to set $\constantCg = \max(\constantCi,\constantCj)$, take the
union bound
over all $k
\le NL-1$ and note that
\begin{equation}
\sum_{k=1}^{NL-1} e^{-(k+1)\ln(eNL/(k+1))} \to 0,
\end{equation}
as $N \to \infty$ for $L \ge 2$.
\end{proof}

\section{Proofs of Propositions \ref{th:asymptotic-MU},\ref{prop:CP-asymptotic},\ref{prop:majorization-weak}}\label{sect:simple-maj}

\begin{proof}[Proof of Proposition~\ref{th:asymptotic-MU}]
Plugging the estimation from Theorem~\ref{th:jiang}
 to the Maassen-Uffink relation, we obtain that with probability tending to one as $N \to \infty$,
\begin{equation}
- \ln c(U)^2 = -\ln \left((1+o(1))\frac{2}{N} \ln N \right)
= \ln N - \ln \ln N - \ln 2 + o(1).
\end{equation}
\end{proof}

\begin{proof}[Proof of Proposition \ref{prop:CP-asymptotic}]
Note that $c(U)^2 + c_2(U)^2 \geq \|\widehat{U}^{(1,2)}\|^2$. Thus, by Theorem~\ref{th:submatrices-fixed-size}, we obtain that for all $\varepsilon > 0$,
\begin{equation}\label{eq:c_2}
\p \left(
c_2^2 \ge (1-\varepsilon) \frac{\ln N}{N}
\right)
\to 1 \text{ as } N \to \infty.
\end{equation}
In particular (again by Theorem \ref{th:submatrices-fixed-size}), with
probability tending to one as $N \to \infty$, $c(U)/c_2(U) \le 3$. Therefore,
with probability tending to one as $N \to \infty$,
\begin{equation}
B_{\textrm{CP}} =
-\ln c(U) ^2+\left( \frac12 -\frac{c(U)}{2} \right) \ln\frac{c(U)^2}{c_2(U)^2}
= - \frac12 (\ln c(U)^2 + \ln c_2(U)^2) + o(1).
\end{equation}
To prove \eqref{eqn:CP-asymptotic} it is now enough to combine the above
estimate with Proposition \ref{th:asymptotic-MU} and \eqref{eq:c_2}.
\end{proof}

\begin{proof}[Proof of Proposition \ref{prop:majorization-weak}]
Denote $q_1 = R_1$, $q_i = R_i - R_{i-1}$ for $i = 2,\ldots,N$. We have $q_1 \ge
\frac{1}{4}$. If $q_1 = 1$, then $H(Q) = 0$, otherwise we have with $I = \{2\le
i \le N\colon q_i > 0\}$,
\begin{align}
H(Q) &= -q_1\ln q_1 + \sum_{i \in I} -q_i\ln q_i\\
& = -q_1\ln q_1 + (1-q_1)\sum_{i\in I} \frac{q_i}{1-q_1}\ln \frac{1}{q_i}\\
& \le -q_1\ln q_1 + (1-q_1)\ln\Big(\sum_{i\in I} \frac{q_i}{1-q_1}\frac{1}{q_i}\Big)\\
& = -q_1\ln q_1 - (1-q_1)\ln(1-q_1) + (1-q_1)\ln |I|\label{eq:almost_last}\\
&\le H(q_1,1-q_1) + (1-q_1)\ln (N-1).\label{eq:D_last}
\end{align}
where in the first inequality we used concavity of the logarithm and the fact
that $\sum_{i\in I} q_i = 1- q_1$. Now, the expression \eqref{eq:D_last}
decreases for $q_1 \in [\frac{1}{N},1]$. In particular for $N \ge 4$ this
implies \eqref{eq:majo-weak}.

\end{proof}


\end{document}